\documentclass[a4paper, 10 pt, conference]{ieeeconf} 

\IEEEoverridecommandlockouts

\usepackage[top=0.75in, bottom=1.44in, left=0.75in, right=0.514in]{geometry}

\usepackage{cite}
\usepackage{graphicx}
\usepackage{amsmath,amssymb,amsfonts}
\usepackage{mathtools}
\usepackage{array}
\usepackage{booktabs}
\usepackage{siunitx}

\usepackage{algorithm}
\usepackage[noend]{algpseudocode}

\usepackage[hidelinks]{hyperref}
\usepackage{url}

\usepackage[capitalise,nameinlink]{cleveref}
\Crefname{equation}{}{}
\Crefname{assumption}{Assumption}{Assumptions}
\Crefname{condition}{Condition}{Conditions}
\Crefname{claim}{Claim}{Claims}
\Crefname{line}{ln}{lns}
\Crefname{section}{\S}{\S}
\Crefformat{section}{\S#2#1#3} 

\usepackage{amsthm}

\usepackage{thm-restate}

\theoremstyle{plain}

\newtheorem{lemma}{Lemma}

\newtheorem{corollary}{Corollary}

\theoremstyle{definition}

\newtheorem{definition}{Definition}
\newtheorem{example}{Example}

\newtheorem{remark}{Remark}

\newtheorem{assumption}{Assumption}

\usepackage{flushend}

\newcommand{\reals}{\mathbb{R}}

\newcommand{\constraints}{\mathcal{X}}
\newcommand{\resource}{B}
\newcommand{\request}{\gamma}
\newcommand{\requests}{\Gamma}
\newcommand{\reqset}{\mathcal{S}}
\newcommand{\distover}{\Delta}
\newcommand{\defeq}{\coloneqq}
\newcommand{\eqdef}{\eqqcolon}
\providecommand{\bigO}{\mathcal{O}}
\providecommand{\pathlength}{\mathrm{P}}
\providecommand{\ev}{\mathbb{E}}

\providecommand{\regret}{\mathsf{Regret}}

\providecommand{\tmax}{\hat{T}}

\newcommand{\fmax}{\bar{f}}
\newcommand{\bmax}{\bar{b}}
\newcommand{\bmin}{\underline{b}}

\providecommand{\nullaction}{\varnothing}

\providecommand{\proj}{\Pi}
\providecommand{\convexset}{\mathcal{K}}
\providecommand{\diameter}{D}
\providecommand{\gradientbound}{G}

\providecommand{\multiplier}{\mu}
\providecommand{\fdual}{f^*}
\providecommand{\dfunc}{\mathsf{D}}

\providecommand{\multmax}{\multiplier_{\max}}
\providecommand{\rolling}{\hat{\multiplier}}
\providecommand{\window}{\omega}
\providecommand{\subgrad}{g}
\providecommand{\avgconsumption}{\rho}
\providecommand{\mdmultiplier}{\phi}
\providecommand{\mdmultset}{\Phi}
\providecommand{\stepsize}{\eta}

\providecommand{\dist}{\mathcal{P}}
\providecommand{\avexdual}{\bar{\dfunc}}

\providecommand{\regretconsti}{C_1}
\providecommand{\regretconstii}{C_2}
\providecommand{\regretconstiii}{C_3}

\providecommand{\crconst}{\beta}
\providecommand{\crconsti}{\regretconsti}
\providecommand{\crconstii}{C_2^\cratio}
\providecommand{\crconstiii}{C_3^\cratio}

\providecommand{\xopt}{x^*}

\providecommand{\onlinecost}{w}
\providecommand{\stoptime}{{T_A}}
\providecommand{\tvoptmult}{\mdmultiplier^\star}

\providecommand{\rollingconst}{\hat{C}}

\providecommand{\egoffset}{\delta}

\newcommand{\cratio}{\alpha}

\providecommand{\sel}{\lambda}

\providecommand{\ts}{s}
\providecommand{\te}{e}

\providecommand{\adv}{\text{\scriptsize\textsf{ADV}}}
\providecommand{\rob}{\text{\scriptsize\textsf{ROB}}}
\providecommand{\alg}{\text{\scriptsize\textsf{ALG}}}
\newcommand{\opt}{\text{\scriptsize\textsf{OPT}}}

\newcommand{\maxrate}{u}
\newcommand{\minrate}{l}
\newcommand{\leading}{\beta}
\newcommand{\remaining}{\bar{B}}
\newcommand{\advremaining}{\bar{B}^\adv}
\newcommand{\algrew}{F}
\newcommand{\algdeg}{B}
\newcommand{\advrew}{F^\adv}
\newcommand{\advdeg}{B^\adv}
\newcommand{\timeremaining}{\bar{T}}

\newcommand{\dindex}{\tau}

\newcommand{\refappendix}[1]{Appendix~\ref{#1}}

\begin{document} 

\title{\LARGE \bf
Robust and learning-augmented algorithms for degradation-aware battery optimization
}

\author{
\authorblockN{Jack Umenberger$^{1}$ and Anna Osguthorpe Rasmussen$^{1,2}$}
\authorblockA{
$^{1}$Department of Engineering Science, University of Oxford, United Kingdom.\\
Email: {\tt\small \{jack.umenberger, anna.osguthorpe\}@eng.ox.ac.uk} \\
$^{2}$Brookfield Renewable North America, United States.\\
}
\thanks{}
}

\maketitle
\thispagestyle{empty}
\pagestyle{empty}
\begin{abstract}
This paper studies the problem of maximizing revenue from a grid-scale battery energy storage system, accounting for uncertain future electricity prices and the effect of degradation on battery lifetime. 
We formulate this task as an online resource allocation problem.
We propose an algorithm, based on online mirror descent, that is no-regret in the stochastic i.i.d. setting and attains finite asymptotic competitive ratio in the adversarial setting (robustness).
When untrusted advice about the opportunity cost of degradation is available, we propose a learning-augmented algorithm that performs well when the advice is accurate (consistency) while still retaining robustness properties when the advice is poor.
\end{abstract}

\section{Introduction}\label{sec:intro}

Grid-scale battery energy storage systems (BESS) face a key operational tradeoff: each charge-discharge cycle has the potential to generate revenue (e.g. from energy arbitrage or ancillary services), but also contributes to degradation, in the form of reduced capacity, increased internal resistance, and reduced round-trip efficiency \cite{collath2022aging}.  
Degradation typically decreases the short-term profitability of the battery (e.g. less capacity for arbitrage, more energy lost to internal resistance), and ultimately leads to retirement of the battery when it reaches end of life (usually between 65\% -- 80\% of its original capacity) \cite{wankmuller2017impact}.
The way in which a battery is operated has a significant impact on the rate of degradation, with high charge/discharge rates, extreme temperatures, and extended periods of time at high state of charge (SoC) all accelerating degradation \cite{maheshwari2020optimizing}.
Consequently, to maximize the revenue of a battery over its entire lifetime, it is essential to balance short-term revenue with the opportunity cost of degradation, i.e. the future revenue that will be lost because of the degradation incurred today.
Unfortunately, the opportunity cost of degradation depends on unknown future quantities, such as electricity prices and battery degradation characteristics.
These quantities are challenging to predict accurately.
Indeed, in the case of electricity prices, it is difficult to estimate even their statistical properties, especially years or decades into the future, due to the evolving nature of electricity grids and markets.
The importance and challenge of accounting for degradation when optimizing BESS operation is widely-appreciated; cf. \Cref{sec:related_work} for a review of the literature.
However, to the best of our knowledge, no prior work has characterized the performance of degradation-aware BESS optimization algorithms when the distributions over future electricity prices and degradation characteristics are unknown, or even adversarially chosen.

\subsection{Contributions}
This paper makes the following contributions:
\begin{enumerate}
    \item We formalize the problem of degradation-aware BESS operation as an online resource allocation problem, incorporating calendar aging (\Cref{ass:calendar}) which is non-standard in the online optimization literature.
    \item We propose an online algorithm (\Cref{alg:robust}) that is no-regret in the stochastic i.i.d. setting (\Cref{thm:regret}) and achieves a finite asymptotic competitive ratio (robustness) in the worst-case/adversarial setting (\Cref{thm:cr}).
    \item When untrusted blackbox advice (predictions) about the opportunity cost of degradation is available, we propose a learning-augmented algorithm (\Cref{alg:learning_augmented}) that achieves arbitrarily close-to-optimal performance when the advice is accurate (consistency), while retaining some of the robustness of \Cref{alg:robust} when the advice is inaccurate. Specifically, we show that \Cref{alg:learning_augmented} is the `most robust algorithm possible' (\Cref{claim:necessary}) that can guarantee total reward within a constant factor $1+\epsilon$ of the reward accrued by following the advice, for arbitrary user-specified $\epsilon>0$ (\Cref{claim:consistent}).
\end{enumerate}
We emphasize that the focus of this paper is on the design and analysis of algorithms for degradation-aware BESS operation. 
We make no contribution to the important topics of battery degradation modelling or electricity price forecasting.

\subsection{Organization}
The remainder of this paper is organized as follows. 
\Cref{sec:preliminaries} introduces the problem formulation and the online resource allocation problem.
\Cref{sec:related_work} reviews related work on both degradation-aware BESS optimization and online optimization, including learning-augmented algorithms.
\Cref{sec:robust} presents our robust algorithm and theoretical guarantees.
\Cref{sec:learning_augmented} presents our learning-augmented algorithm and theoretical guarantees.

\section{Preliminaries}\label{sec:preliminaries}

\subsection{Online resource allocation problem}
An online (single) resource allocation problem is a sequential decision-making problem defined over a finite time horizon of $T$ rounds.
The decision maker (henceforth, \emph{agent}) starts with a finite amount $\resource>0$ of a resource.
At the beginning of each round $t\in[T]$, the agent receives a request $\request_t=(f_t,b_t,\constraints_t)\in\reqset$ comprising a constraint set $\constraints_t\subset\reals^n$, reward function $f_t:\constraints_t\to \reals_+$, and resource consumption function $b_t:\constraints_t\to \reals_+$.
The agent chooses an action $x_t\in\constraints_t$, receives reward $f_t(x_t)$ and consumes $b_t(x_t)$ of the resource.
The goal is to maximize the cumulative reward, subject to the constraint that the total consumption does not exceed the initial quantity $\resource$ of the resource.
We emphasize that the problem is \emph{online} in the sense that the decision at time $t$ is made without knowledge of future requests, $\request_{\tau}$ for $\tau>t$.
The instance ends when either the resource is depleted or when all $T$ rounds have elapsed.
The following assumptions are standard in the literature:
\begin{assumption}
For all requests $\request=(f,b,\constraints)$, we have bounds $\fmax \geq f(x) \geq 0$ and $\bmax \geq b(x) \geq \bmin$ for all $x\in\constraints$, where $\fmax,\bmax,\bmin\geq0$ are known constants.
\end{assumption}

\subsection{Application to BESS optimization}
The online resource allocation problem can be used to model degradation-aware operation of a BESS.
In this application, the resource is the battery's state of health (SoH), which degrades over time as the battery is used.
In a resource allocation problem, the cumulative consumption constraint ${\sum}_{\dindex=1} b_\dindex(x_\dindex)\leq\resource$ is the only factor that explicitly couples decisions at different rounds $t$.
As such, when modelling BESS optimization, each round should correspond to a time period of no less than one day, as other state variables, most signficantly the state of charge (SoC) and battery temperature, will couple (and constrain) decisions within a day.
The total resource $\resource$ corresponds to the total allowable degradation over the battery's lifetime. 
In applications this is typically set to between $20\%$ and $30\%$ of the battery's initial capacity; see \Cref{sec:horizon} below for further discussion.
Each request $\request_t=(f_t,b_t,\constraints_t)$ therefore corresponds to the revenue, degradation, and operational constraints for day $t$. It is worth emphasizing that we make no assumptions about the convexity of these objects; for example, $\constraints_t$ may be mixed-integer, capturing discrete decisions such as whether or not to participate in certain markets.
To capture calendar aging in BESS, we make the following non-standard assumption, not typically made in the online resource allocation literature:
\begin{assumption}[Calendar aging]\label{ass:calendar}
    $b_t(x)\geq\bmin >0$ for all $t$ and $x\in\constraints_t\setminus\{\nullaction\}$, i.e. degradation is strictly positive for all feasible actions, except for a distinguished null action $\nullaction$.
    Moreover, $b_t(0)=\bmin$ for all $t$, i.e. if you `do nothing' you will incur calendar aging. 
\end{assumption}
\begin{assumption}\label{ass:nullaction}
    $\nullaction\in\constraints_t$ only for $t$ such that $\sum_{\dindex=1}^t b_\dindex(\nullaction) = \resource$, i.e. the null action can only be taken when the resource has been fully depleted.
\end{assumption}
\Cref{ass:nullaction} is a simple technicality that ensures that the optimization problem is feasible, despite the presence of calendar aging.
The existence of such a null action $\nullaction$ is standard; however, it is typically assumed that $\nullaction$ is available at all times, allowing the agent to `pass' on a request without incurring any degradation/consumption. This is not the case in the BESS setting.

\subsubsection{Time horizon}\label{sec:horizon}
An Online Resource Allocation problem has a finite number of rounds, $T$, meaning that the agent will observe $T$ requests $\request_1,\cdots,\request_T$. 
Calendar aging implies an upper bound on this time horizon: $T\leq\resource/\bmin\eqdef\tmax$.
However, in practice, $T$ is typically (significantly) smaller than $\tmax$. 
For example, a BESS developer typically secures financing for a limited time period, e.g. 5 years.
It is essential, therefore, that the BESS generates sufficient revenue within this 5 year period, even if the battery could -- if used gently -- last for 15 years. 
We define the quantity $\avgconsumption \defeq \resource/T$. 
For the purpose of theoretical analysis, $\avgconsumption$ is best thought of as a (constant) parameter of the problem instance: given $\avgconsumption$, the total resource $\resource$ scales with $T$ as $\resource=\avgconsumption\cdot T$.
As such, $\avgconsumption$ roughly corresponds to the `wealthiness' of the asset developer: larger $\avgconsumption$ implies that the battery can withstand more degradation over a given time horizon $T$. 
In practice, this is typically achieved through \emph{augmentation} (i.e. adding cells to the battery over time to compensate for capacity lost to degradation) \cite{shin2020optimal} or simply replacing aged cells. Oversizing, i.e. installing a battery with more capacity than needed, is another common strategy to extend battery life \cite{shin2020optimal}. However, oversizing aims to reduce degradation, i.e. decrease $b_t(x_t)$ for given $x_t$, by reducing depth of discharge and current rates (measured in C), whereas augmentation aims to increase $\resource$ without necessarily altering $b_t(x_t)$, and better captures our $\resource=\avgconsumption\cdot T$ model. 

\subsubsection{Limitations of the model} 
As stated earlier, the only state variable that couples decisions at different rounds $t$ is the cumulative resource consumption $\sum_{t=1}^T b_t(x_t)\leq \resource$.
In reality, BESS have other state variables, such a state of charge (SoC) and temperature, that constrain decisions across different rounds (days).
For example, the SoC at the start of day $t$ should equal the SoC at the end of day $t-1$.
This is no problem in practice, as $\constraints_t$ can be defined to enforce such constraints.
It does, however, mean that theoretical results must be interpreted carefully. 
Specifically, when comparing two algorithms (e.g. our online algorithm against the offline optimal), each algorithm should receive the same request $\request_t = (f_t,b_t,\constraints_t)$ at time $t$.
If the different algorithms have different SoCs at the end of day $t-1$, then they will receive different constraints $\constraints_t$ at time $t$, if these constraints enforce SoC continuity.
The simplest `solution' to this problem -- from an analysis perspective -- is to assume that the SoC at the end of each day is constrained some nominal (possibly time-varying) value (or small range of values), e.g. $50\%$, which is common practice in BESS operation.
Reformulating the problem with additional state variables (e.g. SoC and temperature) that couple decisions across rounds -- not only decisions within rounds (days) as in the present formulation -- would considerably complicate the analysis, though is an interesting direction for future work. 

Another apparent limitation of our formulation is that net present value (NPV) is not explicitly modelled through the discounting of future rewards. All our results, for the adversarial model, remain valid when future rewards are discounted. What our formulation does not allow is the exploitation of assumed knowledge of future interest rates, required for NPV calculations.
Given that interests rates are challenging to forecast,\footnote{Figure 1.3 of \cite{stevenson2019impact} shows that the US Federal Reserve consistently underestimated future interest rates from (at least) 2008 to 2016.} we do not consider this a significant limitation, and again emphasize that any such predictions -- accurate or otherwise -- can be incorporated via our learning-augmented algorithm of \Cref{sec:learning_augmented}.

\subsection{Offline solution and dual problem}
Given a sequence of requests $\requests=\lbrace\request_1,\ldots,\request_T\rbrace$,
the optimal solution of the resource allocation problem is denoted: 
\begin{equation}\label{eq:opt}
    \opt(\requests) = \max_{x_t\in\constraints_t} \sum_{t=1}^T f_t(x_t) \quad \text{s.t.} \quad \sum_{t=1}^T b_t(x_t) \leq \resource.
\end{equation}
This is referred to as the offline optimal solution, as it assumes all requests $\requests$ are known in advance.
Let $\multiplier\geq0$ be the dual variable associated with the resource constraint in \Cref{eq:opt}.
Defining 
\begin{equation}
    \fdual_t(\multiplier) = \max_{x\in\constraints_t} f_t(x) - \multiplier\cdot b_t(x),
\end{equation}
the Lagrangian dual function for \Cref{eq:opt} is given by 
\begin{equation}
    \dfunc(\multiplier\mid\requests) = \sum_{t=1}^T \fdual_t(\multiplier) + \multiplier\cdot \resource.
\end{equation}
The dual problem motivates what we will refer to as \emph{opportunity cost} policies:
given a dual variable (opportunity cost) $\multiplier\geq0$, let the decision at time $t$ be given by
\begin{equation}\label{eq:oppcost_policy}
    x_t = \arg\max_{x\in\constraints_t} f_t(x) - \multiplier\cdot b_t(x).
\end{equation}
When strong duality holds (e.g. concave $f$, convex $\constraints$, convex $b$, and Slater's condition), then there exists an optimal $\multiplier^\star$ such that \Cref{eq:oppcost_policy} returns the optimal (offline) solution to \Cref{eq:opt}.
Even when strong duality does not hold (e.g. in the general nonconvex setting), \Cref{eq:oppcost_policy} achieves good performance for appropriately chosen $\multiplier$, as we show in \Cref{sec:robust}.

\subsection{Measures of performance}\label{sec:measures_of_performance}
For theoretical analysis we consider two request models.
In the stochastic i.i.d. model, requests $\request_t$ are drawn i.i.d from a fixed but unknown distribution $\dist\in\distover(\reqset)$, where $\distover(\reqset)$ is the set of all distributions over $\reqset$.
In this setting, we measure the performance of an online algorithm by its regret:
\begin{definition}
    In the stochastic i.i.d. request model, the regret of an online algorithm that produces actions $x_t$ is
 \begin{equation}\label{eq:regret_defn}
    \regret(T) = \sup_{\dist\in\distover(\reqset)}\ev_{\requests\sim\dist}\left[\opt(\requests) - \sum_{t=1}^T f_t(x_t)\right].
\end{equation}   
\end{definition}
An online algorithm is said to be \emph{no-regret} if $\regret(T)$ grows sublinearly in $T$, as this implies that the average regret $\regret(T)/T$ vanishes as $T\to\infty$. 

We also consider the adversarial request model, in which requests $\request_t\in\reqset$ are chosen by an adversary with knowledge of the online algorithm.
In this setting, we measure the performance of an online algorithm by its (asymptotic) competitive ratio:
\begin{definition}
    An online algorithm that produces actions $x_t$ is asymptotically $\cratio$-competitive if
    \begin{equation}
        \lim_{T\to\infty} \sup_{\requests\in\reqset^T} \frac{1}{T}\left(\opt(\requests) - \cratio\cdot\sum_{t=1}^T f_t(x_t)\right) \leq 0.
    \end{equation}
\end{definition}
We remark that asymptotic competitive ratio is weaker than the standard notion of competitive ratio, which requires $\opt(\requests) \leq \cratio\cdot\sum_{t=1}^T f_t(x_t) + \crconst$ for some constant $\crconst$ independent of $T$ and $\requests$.

\section{Related work}\label{sec:related_work}

\subsection{Degradation-aware BESS optimization}
The importance of accounting for degradation when maximizing returns in BESS operation is well-established; excellent summaries of the literature can be found in \cite{maheshwari2020optimizing,collath2022aging}. 
Early works in this area employed proxies to account for the effect of degradation,
constraining or penalizing quantities such as
power, 
daily cycles,
depth of discharge,
minimum/maximum SoC,
or Ah throughput, cf. \cite[\S1.1]{maheshwari2020optimizing} and the references therein.
Later works incorporated explicit models (of varying complexity) for degradation and penalized degradation directly, solving a problem like \Cref{eq:oppcost_policy}, cf. e.g.
\cite{
    koller2013defining,
    fortenbacher2014modeling,
    abdulla2016optimal,
    wankmuller2017impact,
    collath2023increasing,
    kumtepeli2024depreciation,
    nnorom2025aging
},
and the review of \cite{collath2022aging} for further references.\footnote{
Different works use different terminology for the opportunity cost parameter $\multiplier$ in \Cref{eq:oppcost_policy}, e.g. cost of degradation, weighting factor, tradeoff parameter, or marginal cost of discharge.}
We emphasize that the focus of (almost) all of these works is on the formulation and solution of the optimization problem --  in particular, the tradeoff between fidelity of the degradation model and computational tractability -- rather than systematic methods for the selection of the (opportunity) cost of degradation parameter, $\multiplier$, which is typically chosen via a parameter sweep, or set to some proxy for opportunity cost, e.g. the cost of the battery divided by its useful capacity.
The exceptions are \cite{abdulla2016optimal} and \cite{kumtepeli2024depreciation}, which both estimate opportunity cost based on historical revenue and degradation. 
Specifically, \cite{abdulla2016optimal} sets $\multiplier$ to the ratio of cumulative revenue to cumulative degradation, as in \Cref{eq:rolling}, while \cite{kumtepeli2024depreciation} sets $\multiplier$ to the average reward per unit degradation, as in \Cref{eq:ratio_of_averages}. 
Neither work provides any theoretical analysis of the proposed estimation methods. 
 
\subsection{Online resource allocation}
The origins of the online resource allocation problem can, arguably, be traced to online bipartite matching \cite{karp1990optimal}. \emph{AdWords}, a generalization of online bipartite matching, which models the matching of bids by budget constrained advertisers to user queries by internet search engines, attracted much attention \cite{mehta2007adwords, buchbinder2007online, jaillet2011online}. 
The online stochastic convex programming problem, introduced by \cite{agrawal2014fast}, further generalizes AdWords allowing for convex rewards. \cite{agrawal2014fast} provided computationally efficient algorithms that achieve near-optimal regret guarantees in the i.i.d. and random permutation models.
The most relevant prior work is \cite{balseiro2023best}, which studied online mirror descent for general (not necessarily convex) online resource allocation problems in a variety of settings (including both the adversarial and stochastic i.i.d. models). 
We build on the work of \cite{balseiro2023best} by 
(a) incorporating non-zero minimum degradation to model calendar aging (\Cref{ass:calendar}), and 
(b) modifying the mirror descent update to accelerate opportunity cost estimation, without sacrificing theoretical guarantees, cf. \Cref{sec:relation_existing} for further discussion.

\subsection{Learning-augmented algorithms}
Worst-case analysis provides strong performance guarantees, but can be overly pessimistic in practical applications.
One paradigm for moving \emph{beyond worst-case analysis} is the study of learning (or prediction) augmented algorithms \cite[\S30]{roughgarden2021beyond}, which is predicated on the observation that in many practical online problems it is often possible to make good predictions about future requests, particularly given advances in machine learning. 
Roughly speaking, the goal is to design algorithms that perform well when the predictions are accurate (consistency), while still retaining performance guarantees when the predictions are (adversarially) poor (robustness).   
Early work in this area focused on caching \cite{lykouris2018competitive}, job scheduling, and ski rental \cite{purohit2018improving}, but has since expanded to cover a wide variety of online problems, such as bin packing \cite{angelopoulos2023online}, metrical task systems \cite{antoniadis2023online}, set cover \cite{bamas2020primal}, and others.
The most relevant prior work is that which focuses on online knapsack problems, cf. 
\cite{daneshvaramoli2025near} and the references therein for a good overview.
The work of \cite{daneshvaramoli2025near} itself assumes predictions of the minimum value $\hat{v}$ of any item accepted by the optimal offline solution are available, and shows that near-optimal consistency robustness tradeoffs can be achieved by a convex combination of decisions from an advice-following algorithm and the optimal worst-case algorithm of \cite{zhou2008budget}.
In contrast, we assume predictions of the opportunity cost of degradation are available, and design an algorithm that takes the `most robust actions possible' while ensuring consistency with the advice. 
Methodologically, our approach is similar to that of \cite{lechowicz2024chasing}, which studies chasing convex functions with long-term constraints.
Our setting is, in some sense, dual to this problem: the long-term constraint of \cite{lechowicz2024chasing} is of a covering type, while ours is of a packing type.     

\section{Robust algorithm}\label{sec:robust}
This section develops an algorithm for the online (single) resource allocation problem with calendar aging (\Cref{ass:calendar}), as defined in \Cref{sec:preliminaries}.
At time $t$, the algorithm produces a decision $x_t$ based on all requests $\request_1,\ldots,\request_t$ observed up until time $t$.
We refer to \Cref{alg:robust} as \emph{robust} because it attains a finite asymptotic competitive ratio in the adversarial request model, cf. \Cref{sec:adversarial_results}.

\subsection{Additional assumptions and notation}

Let $\remaining_t = \resource - \sum_{\dindex=1}^t b_\dindex(x_\dindex)$ denote the remaining capacity after time $t$.
Let $\multmax$ denote a known upper bound on the opportunity cost, i.e. an upper bound on the optimal Lagrange multiplier corresponding to the resource constraint in \Cref{eq:opt}.
We remind the reader that $\avgconsumption \defeq \resource/T$ is a constant parameter of the problem instance, such that $\resource = \avgconsumption \cdot T$ for the purpose of theoretical analysis (e.g. as $T\to\infty$).
Note that $T\leq\tmax\eqdef\resource/\bmin$ implies that $\avgconsumption \geq \bmin$.

\subsection{A robust algorithm}

The algorithm we propose is listed in \Cref{alg:robust}.
Actions $x_t$ are chosen according to an opportunity-cost-based policy, cf. \Cref{eq:oppcost_policy} and \Cref{eq:robust_action}, using an estimate $\multiplier_t$ of the opportunity cost.
This estimate $\multiplier_t$ is updated at each time using a combination (sum) of a rolling average $\rolling_t$ of past rewards divided by past degradation, cf. \Cref{eq:rolling}, and a correction term $\mdmultiplier_t$ obtained via online gradient descent on the dual function, cf. \Cref{eq:md_update}. 
The correction term is constrained to lie in the set $\mdmultset_t = [-\rolling_t, \multmax - \rolling_t]$ to ensure $\multiplier_t\in[0, \multmax]$.
We discuss the relationship between \Cref{alg:robust} and prior work in \Cref{sec:relation_existing} below, but briefly: online mirror (gradient) descent is known to have good performance in online resource allocation problems \cite{balseiro2023best}, while the addition of a rolling average estimate is designed to accelerate convergence to good opportunity cost estimates, particularly in the the stochastic setting. 
We remark that \Cref{alg:robust} is computationally efficient when the optimization problem in \Cref{eq:robust_action} can be solved efficiently; the rolling average and mirror descent updates are both cheap to compute.

\begin{algorithm}
    \caption{Robust algorithm}\label{alg:robust}
    \begin{algorithmic}[1]
    \Require{Initial opportunity cost estimate $\multiplier_1>0$, 
    and stepsize $\stepsize>0$}
    \For{each time $t = 1, 2, \ldots, T$}
    \State Receive request $\request_t = (f_t, b_t, \constraints_t)$
    \State Compute action
    \begin{equation}\label{eq:robust_action}
        x_t \gets \arg\max_{\substack{x\in\constraints_t \\ b_t(x)\leq\remaining_{t-1}}} f_t(x) - \multiplier_{t}\cdot b_t(x)
    \end{equation}
    \State Update remaining resource: $\remaining_{t} \gets \remaining_{t-1} - b_t(x_t)$
    \State Update rolling average estimate of opportunity cost:
    \begin{equation}\label{eq:rolling}
        \rolling_{t+1} \gets \frac{\sum_{\dindex=1}^t f_\dindex(x_\dindex)}{\sum_{\dindex=1}^t b_\dindex(x_\dindex)}
    \end{equation}    
    \State Obtain sub-gradient of dual function:
    \begin{equation}
        \subgrad_t = \avgconsumption - b_t(x_t)
    \end{equation}
    \State Update dual variable correction:
    \begin{equation}\label{eq:md_update}
        \mdmultiplier_{t+1} \gets \arg\min_{\mdmultiplier\in\mdmultset_{t+1}} \ \stepsize\cdot\subgrad_t\cdot\mdmultiplier + \frac{1}{2}(\mdmultiplier - \mdmultiplier_t)^2
    \end{equation}
    where $\mdmultset_{t+1} = [-\rolling_{t+1}, \multmax - \rolling_{t+1}]$
    \State Update opportunity cost estimate:
    \begin{equation}
        \multiplier_{t+1} \gets \rolling_{t+1} + \mdmultiplier_{t+1}
    \end{equation}
    \EndFor
    \end{algorithmic}
\end{algorithm}

\begin{remark}
    The rolling average estimate in \Cref{eq:rolling} uses all past observations. 
    One can of course use a finite window of past observations, cf. \Cref{eq:ratio_of_averages}.
    This may lead to better performance in practice, but results in worse theoretical guarantees.
    For a finite window size, the path length of $\rolling$ grows linearly, which leads to linear regret in the stochastic i.i.d. request model.
\end{remark}

\subsection{Stochastic i.i.d. request model}

\begin{restatable*}[Regret bound]{theorem}{regretthm}\label{thm:regret}
When requests $\request_t$ are drawn i.i.d. from an unknown distribution, \Cref{alg:robust} with stepsize $\stepsize$ achieves:
\begin{equation*}
    \regret(T) \leq \regretconsti + \regretconstii\cdot\frac{\log(T)+1}{\stepsize} + \regretconstiii\cdot\stepsize\cdot T,
\end{equation*}
where $\regretconsti = \fmax\bmax/\avgconsumption$, 
$\regretconstii = (5\multmax^2\fmax(1+\bmax/\bmin))/(2\bmin)$ 
and $\regretconstiii = (\avgconsumption + \bmax)^2/2$.
\end{restatable*}

A full proof is provided \refappendix{sec:proof_regret}.
Here, we briefly sketch the key ideas.
It is known that the regret can be bounded in terms of the cost of the online dual problem, i.e. $\sum_t \onlinecost_t(\multiplier_t)$ where $\onlinecost_t(\multiplier) = \subgrad_t\cdot\multiplier$, cf. \cite[Theorem 1]{balseiro2023best}.
Standard regret bounds (against a judiciously chosen comparator $\multiplier^\star$) for online mirror descent then lead to the desired result.
Our \Cref{alg:robust} does not apply standard online mirror descent, due to the addition of the rolling average estimate $\rolling_t$.
In particular, with $\multiplier_t = \rolling_t + \mdmultiplier_t$, 
we can express any static comparator $\multiplier^*$ as $\multiplier^* = \rolling_t + \tvoptmult_t$ for $\tvoptmult_t = \multiplier^* - \rolling_t$.
The regret of the online dual problem relative to $\multiplier^*$ is then given by
\begin{equation*}
    {\sum}_{t} \lbrace \onlinecost_t(\multiplier_t) - \onlinecost_t(\multiplier^*) \rbrace = {\sum}_t \lbrace \subgrad_t\cdot\mdmultiplier_t - {\sum}_t \subgrad_t\cdot\tvoptmult_t \rbrace,
\end{equation*}
i.e. the \emph{dynamic regret} of the mirror descent updates $\mdmultiplier_t$ relative to the time-varying comparator $\tvoptmult_t$.
It is known that dynamic regret bounds for mirror descent depend on the path length of the comparator sequence \cite[\S10]{hazan2016introduction}.
In this case, the path length of $\tvoptmult_t$ is equivalent to that of the rolling average estimate $\rolling_t$, which is inherently stable.
We show the path length of $\rolling_t$ is bounded by $\bigO(\log T)$, which gives the desired regret bound.

For appropriately chosen stepsize $\stepsize$, \Cref{thm:regret} yields a sublinear regret bound:

\begin{corollary}\label{cor:regret}
    With stepsize $\stepsize = \sqrt{\log(T)/T}$, \Cref{alg:robust} achieves $\regret(T) = \bigO(\sqrt{T\cdot\log T})$.
\end{corollary}

\begin{proof}
    From \Cref{thm:regret}, we have $\regret(T)$ bounded by
    \begin{align*}
    & \regretconsti + \regretconstii\cdot\sqrt{T}\cdot\frac{\log(T)+1}{\sqrt{\log T}} + \regretconstiii\cdot\frac{\sqrt{\log T}}{\sqrt{T}}\cdot T \\ 
    =& \regretconsti + \regretconstii\cdot\sqrt{T\cdot\log{T}} + \regretconstiii\cdot\sqrt{T\cdot\log T} + \regretconstii\cdot\frac{\sqrt{T}}{\sqrt{\log T}} \\
    =& \bigO(\sqrt{T\cdot\log T}),
    \end{align*}
    as $\sqrt{T}/\sqrt{\log T}$ grows more slowly than $\sqrt{T\cdot\log T}$.
\end{proof}

\subsection{Adversarial request model}\label{sec:adversarial_results}

\begin{restatable*}[Asymptotic competitive ratio]{theorem}{crthm}\label{thm:cr}
    Let $\cratio = (\bmax - \bmin)/(\avgconsumption - \bmin)$.
    When requests $\request_t$ are chosen by an adversary, \Cref{alg:robust} with stepsize $\stepsize$ achieves
\begin{equation*}
    \opt(\requests) - \cratio\cdot\sum_{t=1}^T f_t(x_t) \leq 
    \crconsti + \crconstii\cdot\frac{\log(T)+1}{\stepsize} + \crconstiii\cdot\stepsize\cdot T,
\end{equation*}
where $\crconstii=\cratio\cdot\regretconstii$ and $\crconstiii=\cratio\cdot\regretconstiii$, with $\regretconsti$, $\regretconstii$ and $\regretconstiii$ as in \Cref{thm:regret}.    
\end{restatable*}

A full proof is provided in \refappendix{sec:proof_cr}.
For appropriately chosen stepsize $\stepsize$, this yields an asymptotic competitive ratio of $\cratio$.

\begin{corollary}
    With stepsize $\stepsize = \sqrt{\log(T)/T}$, \Cref{alg:robust} is asymptotically $\cratio$-competitive.
\end{corollary}
\begin{proof}
    From the proof of \Cref{cor:regret}, with stepsize $\stepsize = \sqrt{\log(T)/T}$ we have
    \begin{equation}
        \opt(\requests) - \cratio\cdot{\sum}_{t=1}^T f_t(x_t) \leq \bigO(\sqrt{T\cdot\log T}).
    \end{equation}
    Dividing both sides by $T$, taking the limit $T\to\infty$, and noting that $\sqrt{T\cdot\log T}$ is sublinear gives the result.
\end{proof}

\subsection{Relationship to existing algorithms}\label{sec:relation_existing}

\subsubsection{Online resource allocation}
Our \Cref{alg:robust} is an extension of the online mirror descent algorithm for online resource allocation problems proposed by \cite{balseiro2023best}.
Specifically, \Cref{alg:robust} introduces a rolling average estimate of the opportunity cost, cf. \Cref{eq:rolling}. Theoretically, our main contribution is to show that \Cref{alg:robust}  retains sublinear regret in the stochastic i.i.d. request model (\Cref{thm:regret}) while maintaining a finite asymptotic competitive ratio in the adversarial request model (\Cref{thm:cr}), despite the addition of the rolling average estimate and the challenges introduced by calendar aging (\Cref{ass:calendar}).
The rolling average estimate does not come for free: constants $\regretconstii$ and $\crconstii$ depend on $\multmax$, which is not the case in \cite{balseiro2023best}.
In addition, our asymptotic competitive ratio of $(\bmax - \bmin)/(\avgconsumption - \bmin)$ is worse than $\bmax/\avgconsumption$ from \cite{balseiro2023best}, and blows up as $\avgconsumption \to \bmin$.
However, this appears to be a consequence of the problem specification, namely calendar aging (\Cref{ass:calendar}), rather than our algorithm.

\subsubsection{Degradation-aware BESS optimization}
The work of \cite{kumtepeli2024depreciation} also proposed to estimate the opportunity cost with a rolling average (without a mirror descent correction). Specifically, \cite{kumtepeli2024depreciation} proposed the estimate
\begin{equation}\label{eq:ratio_of_averages}
    \multiplier_{t+1} = \frac{1}{\window} \cdot {\sum}_{\dindex=t-\window+1}^t \frac{f_\dindex(x_\dindex)}{ b_\dindex(x_\dindex)},
\end{equation}
i.e. an `average of ratios' as opposed to the `ratio of averages' in \Cref{eq:rolling}.
While \Cref{eq:ratio_of_averages} can work well in certain settings, it provably converges to poor estimates, even in simple settings, as the example below demonstrates.

\begin{example}\label{ex:stochastic_case}
    Consider an instance of the online resource allocation problem in the stochastic i.i.d. setting. Let $b_t(x) = \egoffset + x$ with $\egoffset=0.01$ and $f_t(x) = a\cdot x$ with $a\sim U[0,1]$ for all $t$.
    Let $\avgconsumption=0.1$ and $T=2000$ giving $\resource=200$.
    We depict a realization of this instance in \Cref{fig:example}, comparing \Cref{alg:robust}, the `ratio of averages' method of \cite{kumtepeli2024depreciation}, and the online mirror descent method of \cite{balseiro2023best}.
    We take $\window=T$ (i.e. use all past observations) for all rolling averages.
    We also choose $\stepsize = 1/\sqrt{T}$ for both \Cref{alg:robust} and mirror descent.
    For this realization, the optimal multiplier is $\multiplier^\star \approx 0.87$.
    It is easy to show that the `ratio of averages' estimate in \Cref{eq:ratio_of_averages} converges to $\sqrt{2}-1\approx 0.41$ as $\egoffset\to 0$ and $T\to\infty$.  
    Therefore, the method of \cite{kumtepeli2024depreciation} consumes the resource too quickly.
    \Cref{alg:robust} achieves higher total reward than mirror descent, as the rolling average estimate in \Cref{eq:rolling} allows it to approach $\multiplier^\star$ more quickly than mirror descent alone. See \Cref{fig:example} for details.
\end{example}

\begin{figure}[htbp]
    \centering
    \includegraphics[width=1.0\columnwidth]{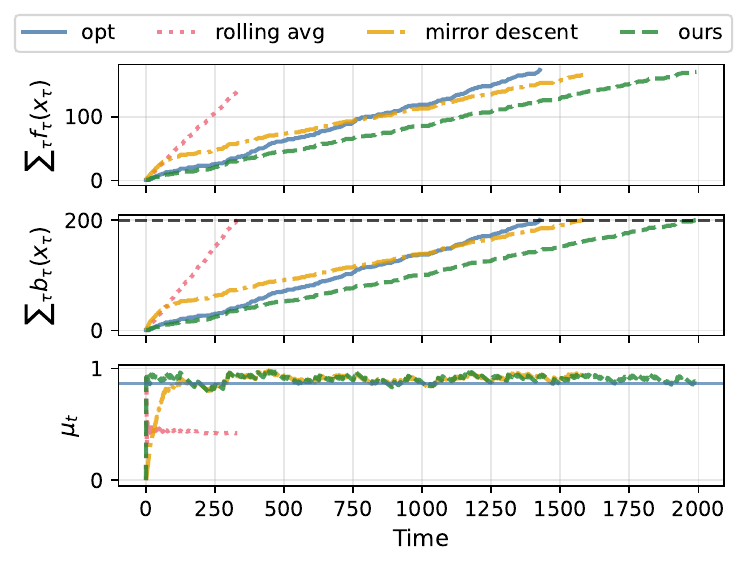}
    \vspace{-1em}
    \caption{Illustration of \Cref{ex:stochastic_case}, showing (top to bottom) cumulative reward, cumulative degradation, and opportunity cost estimates. 
    The optimal total reward and multiplier are $\opt\approx 173.90$ and $\multiplier^\star \approx 0.87$, respectively.
    The rolling average estimate in \Cref{eq:ratio_of_averages} systematically underestimates $\multiplier^\star$, leading to early resource depletion and lower cumulative reward ($140.19$).
    The rolling average estimate in \Cref{eq:rolling} allows our \Cref{alg:robust} to approach $\multiplier^\star$ more quickly than pure mirror descent, resulting in higher total reward ($171.08$ vs $166.10$).
    }
    \label{fig:example}
\end{figure}
\section{Learning-augmented algorithm}\label{sec:learning_augmented}

The robust \Cref{alg:robust} from \Cref{sec:robust} is a purely `backward-looking' algorithm.
In BESS applications, one typically has access to forecasts of future electricity prices, and degradation models, that can be used to predict the optimal opportunity cost $\multiplier^\star$.
Predictions can come from a variety of sources, such as machine learning models trained on historical data, the domain knowledge of a human expert, or some combination thereof.
Unfortunately, such predictions are never perfect.
In this section we assume access to two online algorithms, $\adv$ and $\rob$, that each suggest $\multiplier_t$ at time $t$.
$\adv$ corresponds to untrusted (blackbox) advice, that could be arbitrarily bad.
$\rob$ corresponds to some robust algorithm, such as \Cref{alg:robust}.
We seek a prediction (or \emph{learning}) augmented algorithm that performs arbitrarily close-to-optimal when predictions are accurate, while retaining the robustness of \rob\ even when predictions are poor.
The algorithm (\alg) we propose seeks to take actions that are as `robust as possible' while remaining consistent with \adv, meaning \alg\ is guaranteed to accrue reward no less than $1+\epsilon$ times that of \adv, for arbitrary $\epsilon>0$ (\Cref{claim:consistent}).
We show that the conditions that guarantee consistency are also necessary (\Cref{claim:necessary}).

\subsection{Additional notation and assumptions}

For the advice \adv\ let $x_t^\adv$ be the action corresponding to opportunity cost $\multiplier_t^\adv$ at time $t$,
and denote $F_{\ts:\te}^{\adv} = \sum_{t=\ts}^{\te} f_{t}(x_{t}^{\adv})$ and $B_{\ts:\te}^{\adv} = \sum_{t=\ts}^{\te} b_{t}(x_{t}^{\adv})$.
Let $x_t$, $\algrew_{\ts:\te}$, and $\algdeg_{\ts:\te}$ denote the same quantities for our prediction-augmented algorithm \alg.
Let $\remaining_t = \resource - \algdeg_{1:t}$ and $\advremaining_t = \resource - B^\adv_{1:t}$ be the remaining capacity for $\alg$ and $\adv$ respectively at time $t$.
Let $\leading_t = \algdeg_{1:t} - B^\adv_{1:t}$, i.e. the additional capacity consumed by $\alg$ relative to $\adv$ up to and including time $t$.
\begin{assumption}[Bounds on reward per unit of degradation]\label{ass:bounded-rates}
    For all requests $\request=(f,b,\constraints)$, we have $\minrate\cdot b(x)\leq f(x) \leq \maxrate\cdot b(x)$ for all $x\in\constraints$, where $\maxrate,\minrate\geq0$ are known constants.
\end{assumption}

\begin{remark}
    \Cref{ass:bounded-rates} is stronger than those in \Cref{sec:robust}.
    In particular, if $\minrate>0$ it says that we can always achieve non-zero reward at each $t$.
    We emphasize that our algorithm is still valid when $\minrate=0$; however, the conditions for consistency become much more stringent: specifically, one must ensure that $(1+\epsilon)\cdot\algrew_{1:t} \geq \advrew_{1:t}$ at all $t$, cf. \Cref{eq:non-trivial-time-constraint}, and it is very difficult to achieve $\algdeg_{1:t} \geq B^\adv_{1:t}$, cf. \Cref{eq:leading}. In other words, you are essentially forced to follow the advice.
    This is perhaps unsurprising: in the closely-related online multiple-choice knapsack problem, the best known competitive ratio is $\log(U/L)+2$ where $U$ and $L$ are upper and lower bounds on the density (value divided by weight) of items, respectively \cite{zhou2008budget}.
    The problem becomes intractable as $L\to 0$ \cite{marchetti1995stochastic}.
\end{remark}

\begin{remark}
    Though we assume that \adv\ and \rob\ are opportunity cost policies, it is straightforward to extend the algorithm \alg\ to arbitrary online algorithms by working directly with the decisions $x_t^\adv$ and $x_t^\rob$.
\end{remark}

\subsection{Consistent algorithm}
The learning-augmented algorithm we propose is listed in \Cref{alg:learning_augmented}.
The basic idea behind the algorithm ($\alg$) is simple:
at each time $t$, we simulate \rob\ and \adv\ to generate $\multiplier_t^\rob$ and $\multiplier_t^\adv$, respectively. Then $\alg$ attempts to take the `most robust action possible' subject to constraints that ensure consistency with $\adv$. 
More precisely, we maximize $\sel_t\in[0,1]$ and choose the action $x_t$ according to the opportunity cost policy with $\multiplier_t = \sel_t\cdot\multiplier_t^\rob + (1-\sel_t)\cdot\multiplier_t^\adv$, such that $\sel_t=1$ corresponds to following the robust algorithm, and $\sel_t=0$ corresponds to following the advice. 
To ensure $1+\epsilon$ consistency we enforce the following constraints:
\begin{subequations}\label{eq:non-trivial-time-constraint}
    \begin{flalign}
        &\text{if } \quad\ \ \remaining_t \leq \timeremaining_t\cdot\bmax && \\
        &\text{then } \ (1+\epsilon)\cdot\algrew_{1:t} + \epsilon\cdot\minrate\cdot(\remaining_t - \bmax + \bmin) \geq \advrew_{1:t} &&
    \end{flalign}
\end{subequations}\vspace{-1.5em}
\begin{subequations}\label{eq:trivial-time-constraint}
    \begin{flalign}
        & \text{if } \quad\ \ \remaining_t > \timeremaining_t\cdot\bmax && \\
        & \text{then } \ (1+\epsilon)\cdot\algrew_{1:t} + \epsilon\cdot\minrate\cdot\bmax\cdot\timeremaining_t \geq \advrew_{1:t} &&
    \end{flalign}
\end{subequations}\vspace{-1.5em}
\begin{subequations}\label{eq:leading}
    \begin{flalign}
        & \text{if } \quad\ \ \leading_t>0 && \\
    &\text{then } \  (1+\epsilon)\cdot\algrew_{1:t} + \epsilon\cdot\minrate\cdot\remaining_t \dots && \\
    & \hspace{8em} \dots \geq \advrew_{1:t} + 
     \maxrate\cdot(\min\lbrace \bmax, \remaining_t \rbrace +\leading_t). \nonumber &&
    \end{flalign}
\end{subequations}
Each of these constraints take the form of logical implications, that would typically be encoded with mixed-integer programming techniques, cf. \Cref{sec:computational_considerations} for discussion of computational considerations.
The constraints can be understood as follows.
\Cref{eq:non-trivial-time-constraint} relaxes the consistency constraint by recognizing that, roughly speaking, $\alg$ could follow the actions of $\adv$ for its remaining capacity $\remaining_t$, accumulating at least $\minrate\cdot\remaining_t$ additional reward.
The tightening of \Cref{eq:non-trivial-time-constraint} by $\propto\bmax - \bmin$ ensures that \Cref{eq:trivial-time-constraint} will hold if $\remaining_t > \timeremaining_t\cdot\bmax$.
This latter condition ensures that in a \emph{resource-rich endgame}, where \alg\ cannot possibly exhaust its resource, consistency is maintained even if the worst-case (i.e. smallest) rewards are encountered. 
Finally, when $\leading_t>0$, meaning that $\alg$ has consumed more resource than $\adv$ after time $t$, the term on the RHS of \Cref{eq:leading} tightens the consistency constraint to account for 
the additional reward that \adv\ could accrue with its additional resources.

For the purpose of stating the algorithm, it is useful to define the following \emph{endgame} conditions, which are checked at the start of each time period $t$:
\begin{subequations}\label{eq:endgame}
\begin{align}
    \text{Resource-rich endgame: } & \remaining_{t-1} > \timeremaining_{t-1}\cdot\bmax, \label{eq:resource-rich-endgame}\\
    \text{Resource-poor endgame: } & \remaining_{t-1}\leq\bmax \text{ and } \leading_{t-1}>0  \label{eq:resource-poor-endgame} \\
    \text{No-advice endgame: } & B^\adv_{1:t-1}=\resource \label{eq:advice-exhausted-endgame} \\
    \text{Time's up endgame: } & t=T \label{eq:time-up-endgame}
\end{align}
\end{subequations}
If any of these conditions are satisfied, then $\alg$ sets $\multiplier_\tau=0$ and plays greedily for all remaining time periods $\tau\geq t$.

\begin{algorithm}
    \caption{Learning-augmented algorithm}\label{alg:learning_augmented}
    \begin{algorithmic}[1]
    \Require{Consistency parameter $\epsilon>0$}
    \For{each time $t = 1, \ldots, T$}
    \If{\Cref{eq:resource-rich-endgame} or \Cref{eq:resource-poor-endgame} or \Cref{eq:advice-exhausted-endgame} or \Cref{eq:time-up-endgame}}
    \State $\multiplier_t \gets 0$ 
    \Else  
    \State \vspace{-1.8em}
    \begin{align*}
        \sel_t \gets \arg\max_{\sel\in[0,1]} \ & \sel \quad
        \text{s.t.} \quad 
        \text{\Cref{eq:non-trivial-time-constraint}, }
        \text{\Cref{eq:trivial-time-constraint}, }
        \text{\Cref{eq:leading}} 
    \end{align*}
    \State $\multiplier_t \gets \sel_t\cdot \multiplier_t^\rob + (1-\sel_t)\cdot \multiplier_t^\adv$
    \EndIf
    \State Compute $x_t$ according to \Cref{eq:robust_action} with multiplier $\multiplier_t$   
    \State Update $\leading_t$, $\remaining_t$, $\advremaining_t$, $\timeremaining_t$
    \EndFor
    \end{algorithmic}
\end{algorithm}

We emphasize that \Cref{alg:learning_augmented} simulates \rob\ and \adv\ independently of the actions taken by \alg.
This means that consistency is relative to the actual counterfactual performance of \adv, i.e. $(1+\epsilon)\cdot\algrew_{1:T} \geq \advrew_{1:T}$ implies that we (\alg) are no more than a factor $1+\epsilon$ worse than if we had simply followed the advice at all times.
In the event that \rob\ exhausts its resource (in simulation) before \alg, say at time $\tau$, we simply take $\multiplier_t^\rob=\multiplier_\tau^\rob$ for all $t>\tau$.

\subsection{Computational considerations}\label{sec:computational_considerations}
We emphasize that, at time $t$, the quantities $\algrew_t$ and $\algdeg_t$ are functions of the decision $x_t$, which in turn is a function of $\sel_t$ via the opportunity cost policy \Cref{eq:robust_action} with $\multiplier_t=\sel_t\cdot \multiplier_t^\rob + (1-\sel_t)\cdot \multiplier_t^\adv$.
As such, the constraints \Cref{eq:non-trivial-time-constraint,eq:trivial-time-constraint,eq:leading} are in general non-linear, non-convex mixed-integer in $\sel_t$.
Consequently, solving the optimization problem in line 5 of \Cref{alg:learning_augmented} to arbitrary precision appears computationally daunting.
However, the decision variable $\sel_t$ is one-dimensional and bounded in $[0,1]$, and $\algrew_t$ and $\algdeg_t$ are monotonically decreasing in $\multiplier_t$, which is itself linear in $\sel_t$.
As such, one can find good solutions efficiently, e.g. using bisection search.
We also remark, in passing, that \Cref{eq:robust_action} enforces $b_t(x_t)\leq \remaining_{t-1}$, for arbitrary $\multiplier_t\geq0$. 

\subsection{Algorithm properties}
\Cref{alg:learning_augmented} has the following properties:

\begin{restatable*}[Consistency]{proposition}{consistent}
\label{claim:consistent}
For any $\epsilon>0$,
\Cref{alg:learning_augmented} is $1+\epsilon$-consistent with \adv, i.e., $(1+\epsilon)\cdot\algrew_{1:T} \geq \advrew_{1:T}$.
\end{restatable*}
A full proof is provided in \refappendix{sec:proof_consistency}.

\begin{restatable*}[Necessity]{proposition}{necessary}
\label{claim:necessary}
The constraints \Cref{eq:non-trivial-time-constraint,eq:trivial-time-constraint,eq:leading} are necessary for $1+\epsilon$ consistency.
\end{restatable*}
In other words: any algorithm that fails to satisfy \Cref{eq:non-trivial-time-constraint}, \Cref{eq:trivial-time-constraint}, or \Cref{eq:leading} at some time $t$ can be made to violate consistency by an adversarial choice of future requests, cf. \refappendix{sec:proof_necessity} for proof.
\Cref{claim:necessary} justifies the claim that \Cref{alg:learning_augmented} is the `most robust' consistent algorithm possible, as at each time we try to maximally follow the robust algorithm \rob, subject to constraints that are necessary for consistency.

\section{Conclusions}\label{sec:conclusions}
The stated aim of \cite{kumtepeli2024depreciation} was to ``\emph{establish the beginnings of a more systematic approach for determining battery \emph{[degradation]} cost within the real-time optimizer.}''
We hope this paper has made progress towards this aim.
In particular, we have proposed an online algorithm for estimating the opportunity cost of degradation, and established performance guarantees (in terms of total accumulated reward) in both the stochastic i.i.d. and adversarial settings.
When untrusted blackbox advice (forecasts) about the opportunity cost are available, we have proposed a learning-augmented algorithm that takes the `most robust actions possible' while ensuring arbitrarily close-to-optimal performance when the advice is accurate.
Future work will focus on evaluating the proposed algorithms with realistic battery degradation models and electricity price data.

\bibliographystyle{IEEEtran} 
\bibliography{refs}

\pagebreak 
\appendices 
\onecolumn 
\section{Proofs for robust algorithm}

\subsection{Proof of regret bound}\label{sec:proof_regret}

In this section we present the proof of \Cref{thm:regret}, restated below:
\regretthm

\subsubsection{Bounding regret with online dual problem}\label{sec:regret_proof_existing_results}
In this section we recap known results from \cite[Theorem 1]{balseiro2023best}, which is the starting point for our proof.
Denote $\onlinecost_t(\multiplier) = \multiplier\cdot\subgrad_t = \multiplier\cdot(\avgconsumption - b_t(x_t))$.
Let $\stoptime$ denote the first time such that 
\begin{equation}
    \sum_{t=1}^\stoptime b_t(x_t) + \bmax \geq \resource,
\end{equation}
i.e. the first time at which the online algorithm comes sufficiently close to exhausting the resource.
Recall that each request $\request_t=(f_t,b_t,\constraints_t)$ is drawn i.i.d. from an unknown distribution $\dist$.
We wish to bound the expected regret
\begin{equation}\label{eq:simplest_regret_bound}
    \regret(T)=\ev_{\requests\sim\dist^T}\left[\opt(\requests) - \sum_{t=1}^T f_t(x_t)\right] \leq \ev_{\requests\sim\dist^T}\left[\opt(\requests) - \sum_{t=1}^{\stoptime} f_t(x_t)\right],
\end{equation}
where the inequality follows because we ignore any reward accumulated by the online algorithm after time $\stoptime$ (recall that $f_t(\cdot)\geq0$).
Next, we can bound the expected reward of the online algorithm in terms of the dual. 
Given a multiplier $\multiplier\geq0$, we define the following quantity:
\begin{align}
    \frac{1}{T}\ev_{\requests\sim\dist^T}\left[\dfunc(\multiplier\mid\requests)\right]
    &= \frac{1}{T} \ev_{\requests\sim\dist^T}\left[\sum_{t=1}^T \fdual_t(\multiplier) + \multiplier\cdot\resource\right]  \\
    &= \frac{1}{T} \sum_{t=1}^T \ev_{\request\sim\dist}\left[\fdual(\multiplier)\right] + \multiplier\cdot\avgconsumption \\
    &= \ev_{\request\sim\dist}\left[\fdual(\multiplier)\right] + \multiplier\cdot\avgconsumption \\
    &\defeq \avexdual(\multiplier),
\end{align}
i.e. the (averaged) expected dual objective value for a fixed multiplier $\multiplier$.
For all $t\leq \stoptime$ the action $x_t$ selected by \Cref{alg:robust} will not be constrained by the resource limit, i.e. $b_t(x_t)\leq \remaining_{t-1}$ will not be active, and so by the definition of $x_t$ in \Cref{alg:robust} we have
\begin{equation}\label{eq:f_decomp}
    f_t(x_t) = \fdual_t(\multiplier_t) + \multiplier_t\cdot b_t(x_t) = \fdual_t(\multiplier_t) + \multiplier_t\cdot\avgconsumption - \multiplier_t\cdot (\avgconsumption - b_t(x_t)).
\end{equation}
Given this, it can be shown that 
\begin{align}
    \ev\left[\sum_{t=1}^\stoptime f_t(x_t)\right] &= \ev\left[\sum_{t=1}^\stoptime \avexdual(\multiplier_t)\right] - \ev\left[\sum_{t=1}^\stoptime \multiplier_t\cdot\left(\avgconsumption - b_t(x_t)\right)\right] \nonumber \\
    &\geq \ev\left[\sum_{t=1}^\stoptime \avexdual(\bar{\multiplier})\right] - \ev\left[\sum_{t=1}^\stoptime \multiplier_t\cdot\left(\avgconsumption - b_t(x_t)\right)\right] \nonumber \\
    &= \ev\left[\stoptime\cdot \avexdual(\bar{\multiplier})\right] - \ev\left[\sum_{t=1}^\stoptime \onlinecost_t(\multiplier_t)\right], \label{eq:reward_bound}
\end{align}
where $\bar{\multiplier} = \frac{1}{\stoptime}\sum_{t=1}^\stoptime \multiplier_t$.
Here -- it is worth emphasizing -- that establishing the first equality is more involved than simply taking the expected value of the sum of \Cref{eq:f_decomp}; see \cite[Theorem 1]{balseiro2023best} for details.
The inequality follows by the definition of $\bar{\multiplier}$ and the fact that the dual function is convex.
We also emphasize that $\stoptime$ is a random variable, that depends on the request sequence. 
The final inequality follows from the definitions of $\onlinecost_t(\cdot)$ and $\avexdual(\cdot)$.
Substituting \Cref{eq:reward_bound} into \Cref{eq:simplest_regret_bound} we obtain
\begin{equation}\label{eq:regret_avexdual}
    \regret(T) \leq \ev\left[\opt(\requests) - \stoptime\cdot \avexdual(\bar{\multiplier}) + \sum_{t=1}^\stoptime \onlinecost_t(\multiplier_t)\right].
\end{equation}
Next, we can use weak duality to bound $\opt(\requests)$ in terms of the dual function:
\begin{align}
    \ev_{\requests\sim\dist^T}\left[\opt(\requests)\right] &= \frac{\stoptime}{T}\cdot\ev_{\requests\sim\dist^T}\left[\opt(\requests)\right] + \frac{T-\stoptime}{T}\cdot\ev_{\requests\sim\dist^T}\left[\opt(\requests)\right] \nonumber \\
    & \leq \frac{\stoptime}{T}\cdot\ev_{\requests\sim\dist^T}\left[\dfunc(\bar{\multiplier}\mid\requests)\right] + \frac{T-\stoptime}{T}\cdot T\cdot\fmax \nonumber \\
    &= \stoptime\cdot\avexdual(\bar{\multiplier}) + (T-\stoptime)\cdot\fmax.
\end{align}
Here, the inequality follows from weak duality and the fact that $f_t(\cdot)\leq \fmax$ for all $t$.
The final equality follows from the definition of $\avexdual(\cdot)$.
Substituting this into \Cref{eq:regret_avexdual} we obtain
\begin{equation}\label{eq:regret_in_terms_of_regret}
    \regret \leq \ev\left[(T-\stoptime)\cdot\fmax + \sum_{t=1}^\stoptime \onlinecost_t(\multiplier_t)\right].
\end{equation}

\subsubsection{Dynamic regret bound}\label{sec:regret_proof_dynamic_regret}
From \Cref{eq:regret_in_terms_of_regret} we see that it is necessary to bound the cost of the online dual problem, i.e. $\sum_{t=1}^\stoptime \onlinecost_t(\multiplier_t)$.
The results in \cite{balseiro2023best} use known results on online mirror descent to bound this quantity via bounds on the regret of the online dual problem:
\begin{equation}\label{eq:dual_regret}
    \sum_{t=1}^\stoptime \onlinecost_t(\multiplier_t) - \onlinecost_t(\multiplier)
\end{equation}
for an arbitrary $\multiplier\geq0$ (to be chosen judiciously later).
Our \Cref{alg:robust} does not use standard online mirror (gradient) descent, so these results to not apply directly. Instead, we will bound \eqref{eq:dual_regret} in terms of dynamic regret.
For arbitrary $\multiplier\geq0$ denote $\tvoptmult_t = \multiplier - \rolling_t$.
Then because
\begin{align}
    \sum_{t=1}^\stoptime \onlinecost_t(\multiplier_t) - \onlinecost_t(\multiplier) &= \sum_{t=1}^\stoptime \onlinecost_t(\rolling_t + \mdmultiplier_t) - \onlinecost_t(\rolling_t + \tvoptmult_t) \\
    &= \sum_{t=1}^\stoptime (\rolling_t + \mdmultiplier_t)\cdot\subgrad_t - (\rolling_t + \tvoptmult_t)\cdot\subgrad_t \\
    &= \sum_{t=1}^\stoptime \mdmultiplier_t\cdot\subgrad_t - \tvoptmult_t\cdot\subgrad_t 
\end{align}
we see that \eqref{eq:dual_regret} is equivalent to dynamic regret of the online gradient descent algorithm that generates $\mdmultiplier_t$ against the comparator sequence $\tvoptmult_t$. 
Note from \Cref{lem:ogd_dynamic} that $\mdmultiplier$ is updated via online gradient descent.
Note also that $\mdmultiplier_t\in\mdmultset_t = [-\rolling_t, \multmax - \rolling_t]$ for all $t$, so the diameter of the feasible set is $\diameter = \multmax$.
Also, because $\subgrad_t = \avgconsumption - b_t(x_t)$ we have $|\subgrad_t|\leq \gradientbound = \avgconsumption + \bmax$. 
Therefore, we can apply \Cref{lem:ogd_dynamic} to obtain
\begin{equation}\label{eq:dynamic_regret_bound}
    \sum_{t=1}^\stoptime \onlinecost_t(\multiplier_t) - \onlinecost_t(\multiplier) 
    = \sum_{t=1}^\stoptime \mdmultiplier_t\cdot\subgrad_t - \tvoptmult_t\cdot\subgrad_t
    \leq \frac{5\multmax^2}{2\stepsize}\cdot\pathlength(\tvoptmult) + \frac{\stepsize\cdot \stoptime\cdot(\avgconsumption + \bmax)^2}{2}.
\end{equation}
Next, by definition of $\tvoptmult_t$ we have that the path length of the comparator sequence is equal to the path length of the sequence of rolling averages:
\begin{equation*}
    \pathlength(\tvoptmult) = \sum_{t=1}^{\stoptime-1} |\tvoptmult_{t+1} - \tvoptmult_t| = \sum_{t=1}^{\stoptime-1} |\rolling_{t+1} - \rolling_t| = \pathlength(\rolling).
\end{equation*}

\subsubsection{Bounding the path length of the rolling average}
The rolling average estimate is given by
\begin{equation}
    \rolling_t = \frac{\sum_{\dindex=1}^{t-1} f_\dindex(x_\dindex)}{\sum_{\dindex=1}^{t-1} b_\dindex(x_\dindex)} = \frac{\algrew_{1:t-1}}{\algdeg_{1:t-1}}.
\end{equation}
Observe that
\begin{equation}
    \rolling_{t+1} - \rolling_t = \frac{\algrew_{1:t-1} + f_t(x_t)}{\algdeg_{1:t-1} + b_t(x_t)} - \frac{\algrew_{1:t-1}}{\algdeg_{1:t-1}} = \frac{f_t(x_t)-\rolling_t\cdot b_t(x_t)}{\algdeg_{1:t-1} + b_t(x_t)}
\end{equation}
Next observe that $\algrew_{1:t}\leq t\cdot\fmax$ and $\algdeg_{1:t}\geq t\cdot\bmin$, which implies that
\begin{equation}
    \rolling_t \leq \frac{\fmax}{\bmin}.
\end{equation}
Then
\begin{align}
    |\rolling_{t+1} - \rolling_t| = \frac{|f_t(x_t)-\rolling_t\cdot b_t(x_t)|}{\algdeg_{1:t-1} + b_t(x_t)} \leq \frac{\fmax + \rolling_t\cdot b_t(x_t)}{t\cdot \bmin} \leq \frac{\rollingconst}{t},
\end{align}
where 
\begin{equation}
    \rollingconst = \frac{\fmax + \frac{\fmax}{\bmin}\cdot\bmax}{\bmin} = \frac{\fmax}{\bmin}\cdot\left(1+\frac{\bmax}{\bmin}\right).
\end{equation}
So
\begin{equation}
    \pathlength(\rolling) = \sum_{t=1}^{\stoptime-1} |\rolling_{t+1} - \rolling_t| \leq \rollingconst\cdot\sum_{t=1}^{\stoptime-1} \frac{1}{t} = \rollingconst\cdot H_{\stoptime-1} \leq \rollingconst\cdot(\log(\stoptime-1) + 1) \leq \rollingconst\cdot(\log(T) + 1),
\end{equation}
where $H_n$ is the $n$-th harmonic number, and the final inequality follows from $\stoptime\leq T$. Substiuting this into the dynamic regret bound \Cref{eq:dynamic_regret_bound} gives
\begin{equation}\label{eq:dynamic_regret_bound_final}
    \sum_{t=1}^\stoptime \onlinecost_t(\multiplier_t) \leq \sum_{t=1}^\stoptime \onlinecost_t(\multiplier) + \frac{5\multmax^2}{2\stepsize}\cdot\rollingconst\cdot(\log(T) + 1) + \frac{\stepsize\cdot \stoptime\cdot(\avgconsumption + \bmax)^2}{2}.
\end{equation}

\subsubsection{Putting it all together}
Let us return to the bound in \Cref{eq:regret_in_terms_of_regret} and substitute in the dynamic regret bound from \Cref{eq:dynamic_regret_bound_final}.
\begin{equation}
    \regret(T) \leq \ev\left[(T-\stoptime)\cdot\fmax + \sum_{t=1}^\stoptime \onlinecost_t(\multiplier)\right] + \frac{\regretconstii}{\stepsize}\cdot(\log(T) + 1) + \regretconstiii\cdot\stepsize\cdot T,
\end{equation}
where $\regretconstii = (5\multmax^2\cdot\rollingconst)/2$ and $\regretconstiii = (\avgconsumption + \bmax)^2/2$.
and consider the following two cases.
First, suppose that $\stoptime<T$.
By choosing $\multiplier=\fmax/\avgconsumption$, and recalling that $\sum_{t=1}^\stoptime b_t(x_t) + \bmax\geq T\cdot\avgconsumption$, we have the following:
\begin{equation}\label{eq:multiplier_for_early_stop}
    \sum_{t=1}^\stoptime \onlinecost_t(\multiplier) = \sum_{t=1}^\stoptime \frac{\fmax}{\avgconsumption}\cdot(\avgconsumption - b_t(x_t)) \leq  \fmax\cdot\stoptime + \frac{\fmax}{\avgconsumption}\cdot\left(\bmax - T\cdot\avgconsumption \right) = \frac{\fmax}{\avgconsumption}\cdot\bmax - \fmax\cdot(T-\stoptime).
\end{equation}
Substituting this into the regret bound above gives
\begin{equation}
    \regret \leq \frac{\fmax}{\avgconsumption}\cdot\bmax + \frac{\regretconstii}{\stepsize}\cdot(\log(T) + 1) + \regretconstiii\cdot\stepsize\cdot T,
\end{equation}
which, recognizing that $\regretconsti = \fmax\cdot\bmax/\avgconsumption$, is the desired result.
Finally, suppose that $\stoptime=T$.
In this case, we can simply choose $\multiplier=0$ which renders $\sum_{t=1}^\stoptime \onlinecost_t(0)=0$ and produces the regret bound
\begin{equation}
    \regret \leq \frac{\regretconstii}{\stepsize}\cdot(\log(T) + 1) + \regretconstiii\cdot\stepsize\cdot T,
\end{equation}
which is obviously no larger than the desired bound.

\subsection{Proof of asymptotic competitive ratio}\label{sec:proof_cr}
In this section we present the proof of \Cref{thm:cr}, restated below:
\crthm
The proof follows a similar pattern to that of \Cref{sec:proof_regret}, using techniques from \cite{balseiro2023best}.
As in \Cref{sec:proof_regret}, the key is to show that the desired result holds when we can bound the regret of the online dual problem appropriately.
We then make use of the same dynamic regret bounds as in \Cref{sec:proof_regret}.

\subsubsection{Bounding competitive ratio with online dual problem}
For a fixed sequence of requests $\requests$, let $\xopt_t\in\constraints_t$ denote the offline optimal action of $\opt(\requests)$ at time $t$.
For any $t\leq\stoptime$ we have $x_t=\arg\max_{x\in\constraints_t} f_t(x) - \multiplier_t\cdot b_t(x)$, as the resource constraint is not active.
By the definition of $x_t$ we have
\begin{equation}\label{eq:cf_optimality_inequality}
    f_t(x_t) - \multiplier_t\cdot b_t(x_t) \geq f_t(\xopt_t) - \multiplier_t\cdot b_t(\xopt_t),
\end{equation}
and 
\begin{equation}\label{eq:cf_null_inequality}
    f_t(x_t) - \multiplier_t\cdot b_t(x_t) \geq -\multiplier_t\cdot \bmin,
\end{equation}
where the second inequality follows from $f_t(0)\geq0$ and $b_t(0)= \bmin>0$. 
For any $\cratio\geq1$ we have
\begin{align}
    \cratio\cdot f_t(x_t) &= f_t(x_t) + (\cratio - 1)\cdot f_t(x_t) \nonumber \\
    & \geq f_t(\xopt_t) + \multiplier_t\cdot (b_t(x_t) - b_t(\xopt_t)) + (\cratio - 1)\cdot \multiplier_t\cdot ( b_t(x_t) - \bmin) \pm \cratio\cdot\multiplier_t\cdot\avgconsumption  \nonumber \\
    & = f_t(\xopt_t) - \cratio\cdot\multiplier_t\cdot (\avgconsumption - b_t(x_t)) + \multiplier_t\cdot(\cratio\cdot \avgconsumption - b_t(\xopt_t) - (\cratio - 1)\cdot \bmin) \nonumber \\
\end{align}
Here the first inequality follows from \Cref{eq:cf_optimality_inequality} and \Cref{eq:cf_null_inequality}, and the simultaneous addition and subtraction of $\cratio\cdot\multiplier_t\cdot\avgconsumption$.
The equality is a simple rearrangement.
We can ensure non-negativity of the final term by choosing $\cratio$ such that
\begin{equation}
    \cratio\cdot \avgconsumption  - (\cratio - 1)\cdot \bmin \geq \bmax \geq b_t(\xopt_t) \impliedby \cratio \geq \frac{\bmax - \bmin}{\avgconsumption - \bmin}.
\end{equation}
For this choice of $\cratio$ summing over $t=1,\ldots,\stoptime$ gives
\begin{equation}\label{eq:cr_and_dual_cost}
    \cratio\cdot \sum_{t=1}^\stoptime f_t(x_t) \geq \sum_{t=1}^\stoptime f_t(\xopt_t) - \sum_{t=1}^\stoptime \multiplier_t\cdot (\avgconsumption - b_t(x_t)).
\end{equation}
This is the same inequality as in \cite[(18)]{balseiro2023best}, but with a different (worse) constant $\cratio$, due to the fact that $b_t(0)=\bmin>0$ in our setting.

\subsubsection{Putting it all together}
Recall from \Cref{sec:proof_regret} the notation $\onlinecost_t(\multiplier) = \multiplier\cdot(\avgconsumption - b_t(x_t))$.
For any $\multiplier\geq0$ we have
\begin{align*}
    \opt(\requests) - \cratio\cdot \sum_{t=1}^T f_t(x_t) & \leq \sum_{t=1}^T f_t(\xopt_t) - \cratio\cdot \sum_{t=1}^\stoptime f_t(x_t) \\
    & \leq \sum_{t=\stoptime+1}^T f_t(\xopt_t) + \cratio\cdot\sum_{t=1}^\stoptime \multiplier_t\cdot (\avgconsumption - b_t(x_t)) \\
    & \leq (T-\stoptime)\cdot\fmax + \cratio\cdot\sum_{t=1}^\stoptime \onlinecost_t(\multiplier)
\end{align*}
where the first inequality follows from $\stoptime\leq T$, the second from \Cref{eq:cr_and_dual_cost}, and the final inequality from $f_t(\cdot)\leq \fmax$ for all $t$.
Recall also the dynamic regret bound in \Cref{eq:dynamic_regret_bound_final}, which also holds in the adversarial setting. Specifically, the proof of \Cref{eq:dynamic_regret_bound_final} follows from standard dynamic regret bounds for online gradient descent, and the bound on the path length of the rolling average, which only requires that the requests be suitably bounded. Applying this dynamic regret bound gives
\begin{align}\label{eq:cr_final_bound}
    \opt(\requests) - \cratio\cdot \sum_{t=1}^T f_t(x_t) & \leq (T-\stoptime)\cdot\fmax + \cratio\cdot\sum_{t=1}^\stoptime \onlinecost_t(\multiplier) + \cratio\cdot\left( \frac{5\multmax^2}{2\stepsize}\cdot\rollingconst\cdot(\log(T) + 1) + \frac{\stepsize\cdot \stoptime\cdot(\avgconsumption + \bmax)^2}{2} \right).
\end{align}
As before, we consider two cases.
First, suppose that $\stoptime<T$.
Then by choosing $\multiplier=\fmax/(\cratio\cdot\avgconsumption)$ and repeating the arguments in \Cref{eq:multiplier_for_early_stop} we have
\begin{equation}
    \cratio\cdot \sum_{t=1}^\stoptime \onlinecost_t(\multiplier) \leq \frac{\fmax}{\avgconsumption}\cdot\bmax - \fmax\cdot(T-\stoptime).
\end{equation}
(We emphasize that this is nothing more than \Cref{eq:multiplier_for_early_stop} scaled by $\cratio$.)
Substituting this into \Cref{eq:cr_final_bound} gives
\begin{equation}
    \opt(\requests) - \cratio\cdot \sum_{t=1}^T f_t(x_t) \leq \frac{\fmax}{\avgconsumption}\cdot\bmax + \cratio\cdot\left( \frac{5\multmax^2}{2\stepsize}\cdot\rollingconst\cdot(\log(T) + 1) + \frac{\stepsize\cdot T\cdot(\avgconsumption + \bmax)^2}{2} \right),
\end{equation}
which is the desired result.
Finally, suppose that $\stoptime=T$.
In this case, we can simply choose $\multiplier=0$ which renders $\sum_{t=1}^\stoptime \onlinecost_t(0)=0$ and results in
\begin{equation}
    \opt(\requests) - \cratio\cdot \sum_{t=1}^T f_t(x_t) \leq \cratio\cdot\left( \frac{5\multmax^2}{2\stepsize}\cdot\rollingconst\cdot(\log(T) + 1) + \frac{\stepsize\cdot T\cdot(\avgconsumption + \bmax)^2}{2} \right),
\end{equation}
which is obviously no larger than the desired bound.

\subsection{Standard results on dynamic regret}
Consider the standard online convex optimization setting: at each time $t=1,\ldots,T$ an algorithm selects an action $x_t$ from a convex set $\convexset\subseteq\reals^n$, then observes a differentiable convex loss function $f_t:\convexset\to\reals$ and incurs loss $f_t(x_t)$.
Given a sequence $\lbrace u_t\rbrace_{t=1}^T$ of comparators, define the path length of the comparator sequence as
\begin{equation}
    \pathlength(u) = \sum_{t=1}^{T-1} \|u_{t+1} - u_t\| + 1.
\end{equation}
(The $+1$ means that a time-invariant comparator sequence has path length 1.)
The following result is standard in the online learning literature; see e.g. \cite[Theorem 10.1]{hazan2016introduction}
\begin{lemma}\label{lem:ogd_dynamic}
    Assume that $\|x-y\|\leq \diameter$ for all $x,y\in\convexset$ and $\|\nabla f_t(x)\|\leq \gradientbound$ for all $x\in\convexset$ and $t=1,\ldots,T$. 
    Online gradient descent with stepsize $\stepsize$ achieves the following dynamic regret bound against the comparator sequence $\lbrace u_t\rbrace_{t=1}^T$:
    \begin{equation}
        \sum_{t=1}^T f_t(x_t) - f_t(u_t) \leq \frac{5\diameter^2}{2\stepsize}\cdot\pathlength(u) + \frac{\stepsize\cdot T\cdot\gradientbound^2}{2}
    \end{equation}
\end{lemma}
\begin{proof}
    Let $y_{t+1} = x_t - \stepsize\cdot\nabla f_t(x_t)$ be the unconstrained update,
    and $x_{t+1} = \proj_\convexset(y_{t+1})$ be the projection onto $\convexset$ giving actual decision at time $t+1$. First, observe that 
    \begin{equation}
        \|x_{t+1} - u_{t}\|^2 \leq \|y_{t+1} - u_t\|^2 = \|x_t - u_t\|^2 + \stepsize^2\|\nabla f_t(x_t)\|^2 - 2\stepsize\nabla f_t(x_t)^\top (x_t - u_t),
    \end{equation}
    where the inequality follows from the non-expansiveness of the projection operator.
    Rearranging and using $\|\nabla f_t(x)\|\leq \gradientbound$ gives
    \begin{equation}
        2\cdot\nabla f_t(x_t)^\top (x_t - u_t) \leq \frac{\|x_t - u_t\|^2 - \|x_{t+1} - u_t\|^2}{\stepsize} + \stepsize\cdot\gradientbound^2.
    \end{equation}
    By convexity of $f_t$, and the above inequality, we have
    \begin{equation}\label{eq:ogd_convexity_bound}
        2\cdot\left(f_t(x_t) - f_t(u_t)\right) \leq 2\cdot\left(\nabla f_t(x_t)^\top (x_t - u_t)\right) \leq \frac{\|x_t - u_t\|^2 - \|x_{t+1} - u_t\|^2}{\stepsize} + \stepsize\cdot\gradientbound^2.
    \end{equation}
    Inserting
    \begin{equation}
        \|x_t - u_t\|^2 - \|x_{t+1} - u_t\|^2 = \|x_t\|^2 - \|x_{t+1}\|^2 + 2\cdot u_t^\top (x_{t+1} - x_t)
    \end{equation}
    into \Cref{eq:ogd_convexity_bound} and summing over $t=1,\ldots,T$ gives
    \begin{align}
        2\cdot\sum_{t=1}^T \left(f_t(x_t) - f_t(u_t)\right) &\leq \sum_{t=1}^T \frac{\|x_t\|^2 - \|x_{t+1}\|^2 + 2\cdot u_t^\top (x_{t+1} - x_t)}{\stepsize} + \stepsize\cdot T\cdot\gradientbound^2 \\
        &= \frac{1}{\stepsize}\cdot\left(\|x_1\|^2 - \|x_{T+1}\|^2 + 2\sum_{t=2}^T x_t^\top (u_{t-1} - u_t) + 2u_1^\top x_1 - 2u_T^\top x_{T+1}\right) + \stepsize\cdot T\cdot\gradientbound^2. 
    \end{align}
    Using the bounds $\|x\|^2\leq \diameter^2$ and $|u^\top x|\leq \diameter^2$ for all $x,u\in\convexset$, as well as $|x^\top w|\leq \|x\|\cdot \|w\| \leq \diameter\|w\|$ we have 
    \begin{align}
        \sum_{t=1}^T \left(f_t(x_t) - f_t(u_t)\right) &\leq \frac{5\diameter^2}{2\stepsize} + \frac{\diameter}{\stepsize}\cdot\sum_{t=1}^{T-1} \|u_{t+1} - u_{t}\| + \frac{\stepsize\cdot T\cdot\gradientbound^2}{2} \\
        &\leq \frac{5\diameter^2}{2\stepsize}\cdot\left(1 + \sum_{t=1}^{T-1} \|u_{t+1} - u_{t}\|\right) + \frac{\stepsize\cdot T\cdot\gradientbound^2}{2} \\
        &= \frac{5\diameter^2}{2\stepsize}\cdot\pathlength(u) + \frac{\stepsize\cdot T\cdot\gradientbound^2}{2},
    \end{align}
    where the second inequality requires $\diameter\geq 1$. If $\diameter<1$, then the same bound holds with $\diameter$ replaced by $1$.
\end{proof}
\pagebreak 
\section{Proofs for learning-augmented algorithm}\label{sec:learning_proofs}

\subsection{Proof of recursive feasibility}

To prove consistency of \Cref{alg:learning_augmented}, we first prove that there always exists a feasible action for the algorithm to take, such that it reaches the endgame phase with constraints \Cref{eq:trivial-time-constraint,eq:non-trivial-time-constraint,eq:leading} satisfied. 
\begin{lemma}[Recursive feasibility]\label{lem:recursive}
    At any time $t-1$ assume that the constraints \Cref{eq:non-trivial-time-constraint,eq:trivial-time-constraint,eq:leading} are satisfied, and that none of the endgame conditions in \Cref{eq:endgame} are true, i.e. \alg\ has not yet entered the endgame phase.
    Then the action $x_t=x_t^\adv$ ensures that the constraints \Cref{eq:non-trivial-time-constraint,eq:trivial-time-constraint,eq:leading} are satisfied at time $t$.
\end{lemma}

\subsubsection{Case 1: ample time remaining}
Assume that at time $t-1$ we have \Cref{eq:non-trivial-time-constraint} satisfied, i.e.
\begin{equation*}
    (1+\epsilon)\cdot\algrew_{1:t-1} + \epsilon\cdot\minrate\cdot(\remaining_{t-1} - \bmax + \bmin) \geq \advrew_{1:t-1},
\end{equation*}
with $\remaining_{t-1} \leq \timeremaining_{t-1}\cdot\bmax$.
<advice is feasible>
Assume also that, after playing the feasible action $x_t=x_t^\adv$, we still have 
\begin{equation}\label{eq:still_time}
    \remaining_t = \remaining_{t-1} - \advdeg_t \leq \timeremaining_t\cdot\bmax = (\timeremaining_{t-1}-1)\cdot\bmax.
\end{equation}
With $x_t=x_t^\adv$, \Cref{eq:non-trivial-time-constraint} at time $t$ becomes:
\begin{align*}
    (1+\epsilon)\cdot\algrew_{1:t} + \epsilon\cdot\minrate\cdot(\remaining_{t} - \bmax + \bmin)
    &= (1+\epsilon)\cdot\algrew_{1:t-1} + (1+\epsilon)\cdot\advrew_t + \epsilon\cdot\minrate\cdot(\remaining_{t-1} - \advdeg_t - \bmax + \bmin) \\
    &\geq \advrew_{1:t-1} + \advrew_t + \epsilon\cdot(\advrew_t - \minrate\cdot\advdeg_t) \\
    &\geq \advrew_{1:t},
\end{align*}
which, when taken together with \Cref{eq:still_time}, implies that \Cref{eq:non-trivial-time-constraint} is satisfied at time $t$.
Here, the first inequality follows from satisfaction of \Cref{eq:non-trivial-time-constraint} at time $t-1$, and the final inequality from the definiton of $\minrate$, which implies that $\advrew_t \geq \minrate\cdot\advdeg_{t} $.

\subsubsection{Case 2: time begins to run out}
Once again, assume that at time $t-1$ we have \Cref{eq:non-trivial-time-constraint} satisfied, with $\remaining_{t-1} \leq \timeremaining_{t-1}\cdot\bmax$.
This time, however, assume that
\begin{equation}\label{eq:time_begins_to_run_out}
    \remaining_t = \remaining_{t-1} - \advdeg_t > \timeremaining_t\cdot\bmax = (\timeremaining_{t-1}-1)\cdot\bmax,
\end{equation}
i.e. after playing the feasible action $x_t=x_t^\adv$, we enter the resource-rich endgame phase.
With $x_t=x_t^\adv$, we have 
\begin{align*}
    (1+\epsilon)\cdot\algrew_{1:t} + \epsilon\cdot\minrate\cdot\bmax\cdot\timeremaining_{t}
    &= (1+\epsilon)\cdot\algrew_{1:t-1} + (1+\epsilon)\cdot\advrew_t + \epsilon\cdot\minrate\cdot\bmax\cdot(\timeremaining_{t-1}-1) \\
    &\geq (1+\epsilon)\cdot\algrew_{1:t-1}  + (1+\epsilon)\cdot\advrew_t + \epsilon\cdot\minrate\cdot\bmax\cdot\remaining_{t-1} - \epsilon\cdot\minrate\cdot\bmax   \\
    &= (1+\epsilon)\cdot\algrew_{1:t-1} + (1+\epsilon)\cdot\advrew_t + \epsilon\cdot\minrate\cdot\bmax\cdot(\remaining_{t-1} - \bmax + \bmin)  - \epsilon\cdot\minrate\cdot\bmin \\
    & \geq \advrew_{1:t-1} + \advrew_t + \epsilon\cdot(\advrew_t - \minrate\cdot\bmin) \\
    & \geq \advrew_{1:t},
\end{align*}
which, when taken together with \Cref{eq:time_begins_to_run_out}, implies that \Cref{eq:trivial-time-constraint} is satisfied at time $t$.
Here the first inequality follows from $\remaining_{t-1} \leq \timeremaining_{t-1}\cdot\bmax$ at time $t-1$,
the third inequality follows from satisfaction of \Cref{eq:non-trivial-time-constraint} at time $t-1$, and the final inequality from the definiton of $\minrate$.

\subsubsection{Case 3: \alg\ has consumed more than \adv}

Assume that at time $t-1$ we had $\leading_{t-1}>0$, $\remaining_{t-1}\geq \bmax$, and that \Cref{eq:leading} was satisfied, i.e.
\begin{align*}
    (1+\epsilon)\cdot\algrew_{1:t-1} + \epsilon\cdot\minrate\cdot\remaining_{t-1} 
\geq \advrew_{1:t-1} + \maxrate\cdot(\min\lbrace \bmax, \remaining_{t-1} \rbrace +\leading_{t-1}).
\end{align*}
First observe that $\remaining_{t-1} \geq \bmax$ means that $x_t=x_t^\adv$ will not violate the resource constraint, as $\advdeg_t \leq \bmax$ by definition. 
Next, observe that taking action $x_t=x_t^\adv$ results in $\leading_t = \leading_{t-1}$ as
\begin{equation}\label{eq:leading_still_positive}
    \leading_t = \algdeg_{1:t} - \advdeg_{1:t} = \algdeg_{1:t-1} + \advdeg_t - (\advdeg_{1:t-1} + \advdeg_t) = \leading_{t-1} > 0.
\end{equation}
Note also that $\remaining_t = \remaining_{t-1} - \advdeg_t$ when $x_t=x_t^\adv$.
Finally, note that 
\begin{equation}\label{eq:trivial_inequality}
    \min\{ \bmax, \remaining_{t-1} \} \geq \min\{ \bmax, \remaining_{t-1} - \advdeg_t \} 
\end{equation}
Now, consider \Cref{eq:leading} at time $t$:
\begin{align*}
    (1+\epsilon)\cdot(\algrew_{1:t}) + \epsilon\cdot\minrate\cdot\remaining_{t}
    &= (1+\epsilon)\cdot\algrew_{1:t-1} + (1+\epsilon)\cdot\advrew_t + \epsilon\cdot\minrate\cdot(\remaining_{t-1} - \advdeg_t) \\
    &\geq \advrew_{1:t-1} + \maxrate\cdot(\min\{ \bmax, \remaining_{t-1} \} + \leading_{t-1}) + (1+\epsilon)\cdot\advrew_t - \epsilon\cdot\minrate\cdot\advdeg_t \\
    &\geq \advrew_{1:t-1} + \maxrate\cdot(\min\{ \bmax, \remaining_{t-1} - \advdeg_t \} + \leading_{t-1}) + (1+\epsilon)\cdot\advrew_t - \epsilon\cdot\minrate\cdot\advdeg_t \\
    & = \advrew_{1:t} + \maxrate\cdot(\min\{ \bmax, \remaining_{t} \} + \leading_{t} ) + \epsilon\cdot(\advrew_t - \minrate\cdot\advdeg_t)  \\
    &\geq \advrew_{1:t} + \maxrate\cdot(\min\{ \bmax, \remaining_{t} \} + \leading_{t} ),
\end{align*}
which, when taken with \Cref{eq:leading_still_positive}, implies that \Cref{eq:leading} is satisfied at time $t$.
Here, the first inequality follows from satisfaction of \Cref{eq:leading} at time $t-1$, and the second inequality follows from \Cref{eq:trivial_inequality}. The second-to-last identity is a simple rearrangement, and the final inequality follows from the fact that $\advrew_t \geq \minrate\cdot\advdeg_{t} $ by definition of $\minrate$.

\subsubsection{Feasibility at first time step}
At the beginning of the first time step ($t=1$) we have $\leading_0=0$, so
the action $x_1=x_1^\adv$ is feasible and preserves $\leading_1=0$.
Suppose (unlikely though it may be) that
\begin{equation}\label{eq:b1_big}
    \remaining_1 = \resource - \advdeg_1 > \timeremaining_1\cdot\bmax.
\end{equation}
It is clear that \Cref{eq:trivial-time-constraint} is (trivially) satisfied at time $t=1$:
\begin{align*}
    (1+\epsilon)\cdot\algrew_{1} + \epsilon\cdot\minrate\cdot\bmax\cdot\timeremaining_{1}
    \geq (1+\epsilon)\cdot\advrew_1 \geq \advrew_1,
\end{align*}
so we enter the resource-rich endgame phase immediately.
Alternatively (and more likely), suppose that
\begin{equation}\label{eq:sufficient_initial_time}
    \remaining_1 = \resource - \advdeg_1 \leq \timeremaining_1\cdot\bmax.
\end{equation}
Then considering \Cref{eq:non-trivial-time-constraint} at time $t=1$:
\begin{align*}
    (1+\epsilon)\cdot \algrew_{1} + \epsilon\cdot\minrate\cdot(\remaining_{1} - \bmax + \bmin)
    & = (1+\epsilon)\cdot\advrew_1 + \epsilon\cdot\minrate\cdot(\resource - \advdeg_1 - \bmax + \bmin) \\
    & = \advrew_1 + \epsilon\cdot(\advrew_1 - \minrate\cdot\advdeg_1) + \epsilon\cdot\minrate\cdot(\resource - \bmax + \bmin) \\
    & \geq \advrew_1,
\end{align*}
which, when taken together with \Cref{eq:sufficient_initial_time}, implies that \Cref{eq:non-trivial-time-constraint} is satisfied at time $t=1$.
Notice that the inequality holds when $\resource \geq \bmax - \bmin$, which is true in any realistic problem instance.

\subsection{Proof of consistency}\label{sec:proof_consistency}

In this section we present the proof of \Cref{claim:consistent}, restated below:
\consistent

Consider \Cref{alg:learning_augmented} at the beginning of time $t$. Let us enumerate all possibilities to show that (eventually) consistency will be achieved by the end of the problem instance.

\subsubsection{Case 1: resource-rich endgame}
Assume that after time $t-1$, \alg\ satisfies \Cref{eq:resource-rich-endgame}, i.e.
\begin{equation*}
    \remaining_{t-1} > \timeremaining_{t-1}\cdot\bmax,
\end{equation*}
which implies there is sufficient resource remaining for \alg\ to play greedily for all remaining time steps.
Indeed, by definition of \Cref{alg:learning_augmented}, \alg\ will now play $x_\tau$ given by the opportunity cost policy with $\multiplier_\tau=0$ for all $\tau\geq t$.
Let $\algrew_{t:T}^0$ denote the reward that \alg\ accumulates from time $t$ to the end of the instance at time $T$.
Notice that this is the maximum possible reward that any algorithm could accumulate from time $t$ to time $T$, as this corresponds to greedily maximizing reward without regard to resource consumption.
This reward can be lower bounded by 
\begin{equation}\label{eq:endgame_lower_bound}
    \algrew_{t:T}^0 \geq \minrate\cdot\bmax\cdot\timeremaining_{t-1},
\end{equation}
as the RHS corresponds to worst-case reward instances $f_\tau = \minrate\cdot b_\tau$ for all $\tau\geq t$,
and $\multiplier_\tau=0$ implies that \alg\ will play $b_\tau=\bmax$ for all $\tau\geq t$.
By \Cref{lem:recursive}, we know that \Cref{eq:trivial-time-constraint} is satisfied at time $t-1$, i.e.
\begin{equation*}
    (1+\epsilon)\cdot\algrew_{1:t-1} + \epsilon\cdot\minrate\cdot\bmax\cdot\timeremaining_{t-1} \geq \advrew_{1:t-1}.
\end{equation*}
The total reward accumulated by \alg\ from time $1$ to time $T$ is given by
\begin{align*}
    (1+\epsilon)\cdot\algrew_{1:T} &= (1+\epsilon)\cdot(\algrew_{1:t-1} + \algrew_{t:T}^0) \\
    &\geq (1+\epsilon)\cdot\algrew_{1:t-1} + \epsilon\cdot\minrate\cdot\bmax\cdot\timeremaining_{t-1} + \algrew_{t:T}^0 \\
    &\geq \advrew_{1:T},
\end{align*}
i.e. $(1+\epsilon)$ consistency.
Here the first inequality follows from \Cref{eq:endgame_lower_bound}, and the second inequality follows from the fact that \Cref{eq:trivial-time-constraint} holds at time $t-1$ and because $\algrew_{t:T}^0$ is the maximum possible reward that any algorithm could accumulate from time $t$ to time $T$.

\subsubsection{Case 2: resource-poor endgame} 
Assume that at the end of time $t-1$, \alg\ satisfies \Cref{eq:resource-poor-endgame}, i.e. $\leading_{t-1}>0$ and $\remaining_{t-1}\leq\bmax$, with $\advdeg_{1:t-1}<\resource$.
Let us assume that this is the first time $\leading_{t-1}>0$ and $\remaining_{t-1}\leq\bmax$ has occured.
As $\remaining_{t-1}\leq \bmax \leq \bmax\cdot\timeremaining_{t-1}$ (because $\timeremaining_{t-1}\geq 1$ for all $t\leq T$), we know that \alg\ will consume all remaining resource by time $T$ by playing $\multiplier_\tau=0$ for all $\tau\geq t$. 
Again, let $\algrew_{t:T}^0$ denote the reward that \alg\ accumulates from time $t$ to the end of the instance at time $T$ by playing $\multiplier_\tau=0$ for all $\tau\geq t$, which satisfies
\begin{equation}\label{eq:resource_poor_lower_bound}
    \algrew_{t:T}^0 \geq \minrate\cdot\remaining_{t-1},
\end{equation}
as the RHS corresponds to worst-case reward instances $f_\tau = \minrate\cdot b_\tau$ for all $\tau\geq t$.
As $\leading_{t-1}>0$, by \Cref{lem:recursive} we know that \Cref{eq:leading} holds at time $t-1$, i.e.,
\begin{equation*}
    (1+\epsilon)\cdot\algrew_{1:t-1} + \epsilon\cdot\minrate\cdot\remaining_{t-1} \geq \advrew_{1:t-1} + \maxrate\cdot(\remaining_{t-1}+\leading_{t-1}),
\end{equation*}
where we have used $\min\lbrace \bmax, \remaining_{t-1} \rbrace = \remaining_{t-1}$ by assumption.
Finally, recall that the remaining capacity for $\adv$ is 
\begin{equation*}
    \advremaining_{t-1} = \resource - \algdeg_{1:t-1} + \algdeg_{1:t-1} - \advdeg_{1:t-1} = \remaining_{t-1} + \leading_{t-1}.
\end{equation*}
Then the total reward accumulated by \alg\ from time $1$ to time $T$ is given by
\begin{align*}
    (1+\epsilon)\cdot\algrew_{1:T} &= (1+\epsilon)\cdot(\algrew_{1:t-1} + \algrew_{t:T}^0) \\
    &\geq (1+\epsilon)\cdot\algrew_{1:t-1} + \epsilon\cdot\minrate\cdot\remaining_{t-1} \\
    &\geq \advrew_{1:t-1} + \maxrate\cdot(\remaining_{t-1}+\leading_{t-1}) \\
    &= \advrew_{1:t-1} + \maxrate\cdot\advremaining_{t-1} \\
    &\geq \advrew_{1:t-1} + \advrew_{t:T} \\
    &= \advrew_{1:T},
\end{align*}
i.e. $(1+\epsilon)$ consistency.
Here the first inequality follows from \Cref{eq:resource_poor_lower_bound}, the second inequality follows from the fact that \Cref{eq:leading} holds at time $t-1$, the equality follows from the definition of $\advremaining_{t-1}$, and the final inequality follows 
from the fact that $\maxrate\cdot\advremaining_{t-1}$ is an upper bound on the reward that \adv\ could accumulate from time $t$ to time $T$ with remaining capacity $\advremaining_{t-1}$.

\subsubsection{Case 3: no-advice endgame}
Assume that after time $t-1$, \Cref{eq:advice-exhausted-endgame} is satisfied, i.e. we have $\advdeg_{1:t-1}=\resource$.
Note that $\leading_{t-1} \leq 0$, as $\leading_{t-1} > 0$ would imply that $\algdeg_{1:t-1} > \advdeg_{1:t-1} = \resource$, which is impossible.        
Next, observe that the remaining capacity is
\begin{equation*}
    \remaining_{t-1} = \resource - \algdeg_{1:t-1} = \resource - \advdeg_{1:t-1} + \advdeg_{1:t-1} - \algdeg_{1:t-1} = -\leading_{t-1}
\end{equation*}
by definition. 
We may assume that 
\begin{equation}\label{eq:will_exhaust_resource}
    \remaining_{t-1} \leq \timeremaining_{t-1}\cdot\bmax,
\end{equation}
otherwise Case 1 (resource-rich endgame) would apply.
Then, by \Cref{lem:recursive}, we know that \Cref{eq:non-trivial-time-constraint} holds at time $t-1$, i.e.
\begin{equation*}
    (1+\epsilon)\cdot\algrew_{1:t-1} + \epsilon\cdot\minrate\cdot(\remaining_{t-1} - \bmax + \bmin) \geq \advrew_{1:t-1}.
\end{equation*}
Playing $\multiplier_\tau=0$ for all $\tau\geq t$ ensures that \alg\ will exhaust its remaining resource budget by the end of the instance at time $T$, by \Cref{eq:will_exhaust_resource}.
Again, let $\algrew_{t:T}^0$ denote the reward that \alg\ accumulates from time $t$ to the end of the instance at time $T$ by playing $\multiplier_\tau=0$ for all $\tau\geq t$,
and note that
\begin{equation}\label{eq:exhaust_lower_bound}
    \algrew_{t:T}^0 \geq \minrate\cdot\remaining_{t-1},
\end{equation}
as the RHS corresponds to worst-case reward instances $f_\tau = \minrate\cdot b_\tau$ for all $\tau\geq t$.
Then the total reward accumulated by \alg\ from time $1$ to time $T$ is given by
\begin{subequations}\label{eq:consistency_advice_exhausted}
    \begin{align}
        (1+\epsilon)\cdot\algrew_{1:T} &= (1+\epsilon)\cdot(\algrew_{1:t-1} + \algrew_{t:T}^0) \\
        &\geq (1+\epsilon)\cdot\algrew_{1:t-1} + \epsilon\cdot\minrate\cdot(\remaining_{t-1}) + \algrew_{t:T}^0 \\
        &\geq (1+\epsilon)\cdot\algrew_{1:t-1} + \epsilon\cdot\minrate\cdot(\remaining_{t-1}) - \epsilon\cdot\minrate\cdot(\bmax - \bmin) + \algrew_{t:T}^0 \\
        &\geq \advrew_{1:t-1} + \algrew_{t:T}^0 \\
        &\geq \advrew_{1:T},
    \end{align}
\end{subequations}
i.e. $(1+\epsilon)$ consistency.
Here the first inequality follows from \Cref{eq:exhaust_lower_bound}, the second inequality follows because $\bmax - \bmin \geq 0$, the third inequality follows from the fact that \Cref{eq:non-trivial-time-constraint} holds at time $t-1$, 
and the final inequality follows because $\advrew_{1:t-1} = \advrew_{1:T}$, as \adv\ can accumulate no additional reward after time $t-1$ having exhausted its resource budget.
Note that the above arguments go through unchanged even if $\remaining_{t-1}=0$, i.e. \alg\ has also fully consumed its resource budget by time $t-1$.

\begin{remark}
    It may not be necessary for \alg\ to play greedily (i.e. $\multiplier_\tau=0$) from time $t$ to time $T$ to achieve consistency in this case.
    For example,
    \adv\ could continue to follow the \rob\ algorithm for some time. 
    The risk is that \alg\ may fail to accumulate at least $\minrate\cdot\remaining_{t-1}$ reward from time $t$ to time $T$, which can only occur if \alg\ fails to fully consume its remaining resource budget by time $T$.
    Taking $\multiplier_\tau=0$ for all $\tau\geq t$ is a simple way to guarantee that \alg\ fully consumes its remaining resource budget by time $T$, given that \Cref{eq:advice-exhausted-endgame} has been satisfied before \Cref{eq:resource-rich-endgame}.
\end{remark}

\subsubsection{Case 4: time's up endgame}
Assume that $t=T$.
First, we may assume that 
\begin{equation*}
    \remaining_{T-1} \leq \bmax\cdot\timeremaining_{T-1} = \bmax,
\end{equation*}
otherwise, if $\remaining_{T-1} > \bmax$, Case 1 (resource-rich endgame) would apply.
We may also assume that $\leading_{T-1}\leq0$.
Otherwise, because $\remaining_{T-1}\leq \bmax$, if $\leading_{T-1}>0$ then Case 2 (resource-poor endgame) would apply.
Finally, we may assume that $\advdeg_{1:T-1}<\resource$, otherwise Case 3 (\adv\ exhausts resource) would apply.
The remaining analysis is (almost) identical to Case 3.
As $\remaining_{T-1} \leq \bmax\cdot\timeremaining_{T-1}$, by \Cref{lem:recursive} we know that \Cref{eq:non-trivial-time-constraint} holds at time $T-1$,
and by playing $\multiplier_T=0$ \alg\ will accumulate $\algrew_T^0 \geq \minrate\cdot\remaining_{T-1}$ reward at the final time step.
Repeating the arguments in \Cref{eq:consistency_advice_exhausted} leads to 
\begin{equation*}
    (1+\epsilon)\cdot\algrew_{1:T} \geq \advrew_{1:T-1} + \algrew_T^0 \geq \advrew_{1:T},
\end{equation*}
i.e. $(1+\epsilon)$ consistency.
Here the first inequality is \Cref{eq:consistency_advice_exhausted}.
The second inequality follows from the fact that $\algrew_T^0$ is the maximum possibly reward that can be attained with capacity $\remaining_{T-1}$ at time $T$, and $\advremaining_{T-1} = \remaining_{T-1} + \leading_{T-1} \leq \remaining_{T-1}$, as $\leading_{T-1}\leq0$.

\subsubsection{Case 5: none of the above}
Cases 1 to 4 cover all endgame scenarios \Cref{eq:resource-rich-endgame,eq:resource-poor-endgame,eq:advice-exhausted-endgame,eq:time-up-endgame}, respectively.
If none of these cases apply at time $t-1$, then \alg\ will continue to line 5 of \Cref{alg:learning_augmented}.
\Cref{lem:recursive} guarantees that a feasible action exists.
\alg\ will continue taking actions according to line 5 and 6, until one of the endgame conditions occurs.

\subsection{Proof of necessity of constraints}\label{sec:proof_necessity}
In this section we present the proof of \Cref{claim:necessary}, restated below:
\necessary
To see this, let us consider each constraint in turn.

\subsubsection{Violation of \Cref{eq:leading}}

Suppose that after some arbitrary time $t$ we have violated \Cref{eq:leading}, i.e. we have $\leading_{t} > 0$ but
\begin{equation*}
    (1+\epsilon)\cdot\algrew_{1:t} + \epsilon\cdot\minrate\cdot\remaining_{t} = \advrew_{1:t} + \maxrate\cdot(\remaining_{t}+\leading_{t}) - \eta,
\end{equation*}
for some strictly positive $\eta>0$.
Note that we have here assumed that $\min\lbrace \bmax, \remaining_{t} \rbrace = \remaining_{t}$. 
Now suppose that the adversary chooses worst-case (minimum reward) requests until \alg\ exhausts its capacity, and then switches to best-case (maximum reward) requests for the remainder of the instance.
In this case, \alg\ will accumulate $\algrew_{t+1:T}=\minrate\cdot\remaining_{t}$ reward from time $t+1$ to time $T$, while \adv\ will accumulate $\advrew_{t+1:T}=\minrate\cdot\remaining_{t} + \maxrate\cdot\leading_{t}$ reward from time $t+1$ to time $T$, using its additional capacity $\leading_{t} > 0$. In this case, the total reward accumulated by \alg\ is:
\begin{align*}
    (1+\epsilon)\cdot\algrew_{1:T} & = (1+\epsilon)\cdot(\algrew_{1:t} + \algrew_{t+1:T}) \\
    & = (1+\epsilon)\cdot\algrew_{1:t} + \epsilon\cdot\minrate\cdot\remaining_{t} + \minrate\cdot\remaining_{t} \\
    & = \advrew_{1:t} + \maxrate\cdot(\remaining_{t}+\leading_{t}) - \eta + \minrate\cdot\remaining_{t} \\
    & = \advrew_{1:t} + \minrate\cdot\remaining_{t} + \maxrate\cdot\leading_{t} + \maxrate\cdot\remaining_t - \eta \\
    & = \advrew_{1:T} + \maxrate\cdot\remaining_t - \eta,
\end{align*}
where the second identity uses $\algrew_{t+1:T}$, the third identity uses violation of \Cref{eq:leading} at time $t$, and the final identity uses $\advrew_{t+1:T}$.
Now observe that for any given $\eta>0$, the adversary can design an instance such that $\remaining_t$ is arbitrarily small, specifically such that $\maxrate\cdot\remaining_t < \eta$.
Then consistency will be violated. 

\subsubsection{Violation of \Cref{eq:trivial-time-constraint}}

Suppose that after some arbitrary time $t$ we have violated \Cref{eq:trivial-time-constraint}, i.e. we have $\remaining_{t} > \timeremaining_{t}\cdot\bmax$ but
\begin{equation}
    (1+\epsilon)\cdot\algrew_{1:t} + \epsilon\cdot\minrate\cdot\bmax\cdot\timeremaining_{t} < \advrew_{1:t}.
\end{equation}
Suppose the adversary selects requests with `worst-case' rewards $f_\tau(x) = \minrate\cdot b_\tau(x)$ for all $\tau>t$. 
Then, conditioned on this state, the maximum possible total reward that \alg\ can accumulate is
\begin{align*}
    (1+\epsilon)\cdot\algrew_{1:T} & = (1+\epsilon)\cdot(\algrew_{1:t} + \algrew_{t+1:T}) \\
    & \leq (1+\epsilon)\cdot\algrew_{1:t} + (1+\epsilon)\cdot\minrate\cdot\bmax\cdot\timeremaining_{t} \\
    & < \advrew_{1:t} + \minrate\cdot\bmax\cdot\timeremaining_{t} \\
    & = \advrew_{1:T},
\end{align*}
i.e. violation of $(1+\epsilon)$ consistency.
Here the first inequality follows from the fact that with these worst-case rewards, no algorithm can accumulate more than $\minrate\cdot\bmax$ reward at each remaining time step, and the second inequality follows from violation of \Cref{eq:trivial-time-constraint} at time $t$.
The final identity assumes that \adv\ accumulates maximum reward $\minrate\cdot\bmax\cdot\timeremaining_{t}$ for the remainder of the instance.

\subsubsection{Violation of \Cref{eq:non-trivial-time-constraint}}

The necessity of \Cref{eq:non-trivial-time-constraint} stems from the necessity of \Cref{eq:trivial-time-constraint}, as shown above.
Suppose that after arbitrary time $t$ we have violated \Cref{eq:non-trivial-time-constraint}, in particular suppose that we have 
\begin{equation}\label{eq:time_tight}
    \remaining_{t} = \timeremaining_{t}\cdot\bmax,
\end{equation}
but 
\begin{equation}\label{eq:violate_non_trivial}
    (1+\epsilon)\cdot\algrew_{1:t} + \epsilon\cdot\minrate\cdot(\remaining_{t} - \bmax + \bmin) = \advrew_{1:t} - \eta,
\end{equation}
for some strictly positive $\eta>0$.
Let us also assume $\timeremaining_{t}\gg 1$, such that $\remaining_{t} > \bmax$.
Notice that
\begin{equation}\label{eq:comparison_inequality}
    \advrew_{1:t} + \maxrate\cdot(\bmax + \leading_{t}) > 
    \advrew_{1:t} + \epsilon\cdot\minrate\cdot(\bmax - \bmin), 
\end{equation}
unless $\epsilon$ is chosen to be unreasonably large.
As we have assumed \Cref{eq:violate_non_trivial}, the inequality in \Cref{eq:comparison_inequality} implies that \Cref{eq:leading} cannot be satisfied at time $t$.
This implies that we must have $\leading_{t}\leq0$. 
Let us assume that $\leading_t=0$.  
Now suppose that the adversary selects a request at time $t+1$ such that the advice takes an action that consumes resource $\advdeg_{t+1}=\bmin$ and generates reward $\advrew_{t+1}=\minrate\cdot\bmin$.
By the same argument at time $t$, we must also have $\leading_{t+1}\leq 0$, otherwise with $\leading_{t+1}>0$ we would violate \Cref{eq:leading} at time $t+1$. 
As \adv\ consumes the minimum possible amount of resource, to ensure $\leading_{t+1}\leq0$ we must have \alg\ select the same action such that $\algdeg_{t+1}=\advdeg_{t+1}=\bmin$.
This implies that $\leading_{t+1}=\leading_t=0$.
Now that we have established that any action at time $t+1$ must consume exactly $\bmin$ resource, we have
\begin{equation*}
    \remaining_{t+1} = \remaining_{t} - \bmin = \timeremaining_{t}\cdot\bmax - \bmin = (\timeremaining_{t+1}+1)\cdot\bmax - \bmin > \timeremaining_{t+1}\cdot\bmax,
\end{equation*}
i.e. at time $t+1$ we have entered the resource-rich endgame. 
From Case 2 above, we know that \Cref{eq:trivial-time-constraint} must be satisfied at time $t+1$.
However, we have
\begin{align*}
    (1+\epsilon)\cdot\algrew_{1:t+1} + \epsilon\cdot\minrate\cdot\bmax\cdot\timeremaining_{t+1} 
    &= (1+\epsilon)\cdot\algrew_{1:t} + \epsilon\cdot\minrate\cdot\bmax\cdot(\timeremaining_{t}-1) + (1+\epsilon)\cdot\advrew_{t+1} \\
    &= (1+\epsilon)\cdot\algrew_{1:t} + \epsilon\cdot\minrate\cdot(\remaining_{t} - \bmax) \pm \epsilon\cdot\minrate\cdot\bmin + (1+\epsilon)\cdot\advrew_{t+1} \\
    &= (1+\epsilon)\cdot\algrew_{1:t} + \epsilon\cdot\minrate\cdot(\remaining_{t} - \bmax + \bmin) + \advrew_{t+1} + \epsilon\cdot(\advrew_{t+1} - \minrate\cdot\bmin) \\
    &= \advrew_{1:t} - \eta + \advrew_{t+1} + 0 \\
    &= \advrew_{1:t+1} - \eta \\
    &< \advrew_{1:t+1},
\end{align*}
i.e. \Cref{eq:trivial-time-constraint} is violated at time $t+1$, which from Case 2 above, will result in violation of consistency.
Here, the second identity follows from \Cref{eq:time_tight}, the fourth from violation of \Cref{eq:non-trivial-time-constraint} at time $t$ in \Cref{eq:violate_non_trivial} and the fact that $\advrew_{t+1} = \minrate\cdot\bmin$.
\pagebreak

\end{document}